\documentclass{article}

\usepackage{arxiv}

\usepackage[utf8]{inputenc} % allow utf-8 input
\usepackage[T1]{fontenc}    % use 8-bit T1 fonts
\usepackage{hyperref}       % hyperlinks
\usepackage{url}            % simple URL typesetting
\usepackage{booktabs}       % professional-quality tables
\usepackage{amsfonts}       % blackboard math symbols
\usepackage{nicefrac}       % compact symbols for 1/2, etc.
\usepackage{microtype}      % microtypography
\usepackage{lipsum}
\usepackage{amsthm}

\usepackage{graphicx}
\usepackage{amssymb}
\newcommand{\wenbo}[1]{{\color{blue}{\bf\sf [WB: #1]}}}

\usepackage{multirow}

\usepackage{paralist}

\newtheorem{theorem}{Theorem}
\newtheorem{lemma}{Lemma}

\usepackage{indentfirst}
\usepackage{verbatim}
\usepackage{float}
\usepackage[ruled,vlined]{algorithm2e}
\usepackage{url}
\usepackage{amsmath}  
\usepackage{diagbox}
\usepackage{booktabs}
\graphicspath{{figures/}}
\AtBeginDocument{%
	\providecommand\BibTeX{{%
			\normalfont B\kern-0.5em{\scshape i\kern-0.25em b}\kern-0.8em\TeX}}}
\usepackage{menukeys}

\usepackage{subfigure}

\title{Dynamic Window-level Granger Causality of Multi-channel Time Series}

\author{
	Zhiheng Zhang \\
	Beijing University of Posts and Telecommunications\\
	\texttt{studyzzh@163.com}
	\And
	Wenbo Hu \\
	Tsinghua University \\
	\texttt{wenbohu@tsinghua.edu.cn}
	\And 
	Tian Tian \\
	Tsinghua University \\
	\texttt{tian-tian@tsinghua.edu.cn}
	\And 
	Jun Zhu \\
	Tsinghua University \\
	\texttt{dcszj@tsinghua.edu.cn}
}

\begin{document}
	\maketitle
	
	\begin{abstract}
		Granger causality method analyzes the time series causalities without building a complex causality graph.
		However, the traditional Granger causality method assumes that the causalities lie between time series channels and remain constant, which cannot model the real-world time series data with dynamic causalities along the time series channels.
		In this paper, we present the dynamic window-level Granger causality method (DWGC) for multi-channel time series data. We build the causality model on the window-level by doing the F-test with the forecasting errors on the sliding windows. We propose the causality indexing trick in our DWGC method to reweight the original time series data. Essentially, the causality indexing is to decrease the auto-correlation and increase the cross-correlation causal effects, which improves the DWGC method. Theoretical analysis and experimental results on two synthetic and one real-world datasets show that the improved DWGC method with causality indexing better detects the window-level causalities.
		
	\end{abstract}

	% keywords can be removed
	%\keywords{First keyword \and Second keyword \and More}

	\section{Introduction}
	
	%\indent
	\noindent Time series data is the data with the pre-defined time or sequencial order~\cite{hamilton1994time} and is widely used in multifarious real-world applications, such as  signal processing~\cite{scharf1991statistical}, economics~\cite{granger2014forecasting}, control theory~\cite{box2015time}, etc. 
	Typical tasks for time series data include indexing, clustering, classification, and regression~\cite{keogh2003need}. %Various methods have been developed for time series data tasks, including the Box-Jenkins method~\cite{box2015time}, the hidden Markov model~\cite{zucchini2017hidden} and the recurrent neural network~\cite{connor1994recurrent}.
	%At present, the interpretability of deep learning is still to be improved, where causal reasoning becomes the key to break through this bottleneck\cite{sanders2004defense}.\wenbo{what is \cite{sanders2004defense}?} According to professor Judea Pearl, causal reasoning is at the top level of cognitive logic, getting rid of the nature of curve fitting for traditional deep learning (at the bottom level) .\\
	Among the time series tasks, the causal reasoning task is at the top level of cognitive reasoning~\cite{pearl2018theoretical} and is getting rid of the nature of the correlation fitting for traditional statistical machine learning and deep learning, which is nevertheless at the bottom level.
	Due to the natural implication of the temporal precedence, the time-series data therein encapsulates both empirical experiences of the trends, but also the prior knowledge of causalities between different channels~\cite{eichler2013causal}.
	
	Granger~\cite{granger1969investigating} first used the statistical hypothesis test to decide whether one time series channel is useful to predict another, which is known as \emph{Granger causality}~(GC) and widely used in various applications. 
	Lynggarrd and Walther~\cite{lynggarrd1993dynamic} proposed the dynamic interaction models based on the classical `LWF Markov property' for chain graphs~\cite{lauritzen1989graphical,frydenberg1990chain}.
	Pearl and Robins~\cite{pearl1995probabilistic} put forward the `back door' temporal causal conditions and extend the traditional Granger causality to the temporal sequences.   Dahlhaus and Eicher~\cite{dahlhaus2003causality} discussed an alternative approach that defined the graph according to the `AMP Markov property' of \cite{andersson2001alternative}. Eichler~\cite{eichler2012graphical} adopted the mixed graph constraints, derived from ordinary time series, and used a single vertex and directed edges to represent the component series and the causal relationship respectively.
	
	The traditional Granger causality methods are limited to the predictive capability of the autoregressive model (AR) which only used the linear models. 
	Eichler~\cite{eichler2013causal} re-described the problems when transferring the Granger causality to the nonlinear version: 1) the aggregation of the time-varying coefficients overtime required by Granger causality tests, and 2) the instability of the causal structure.
	Chen et al.~\cite{chen2004analyzing} used a delay embedding to get an extended nonlinear version of Granger Causality. 
	Later, Sun~\cite{sun2008assessing} proposed to use the RKHS kernel embeddings to get the nonlinearity.
	%\zhiheng{explains it in more detail and sets the stage for our model.}
	%\wenbo{this method uses windows-sliding causality?\zhiheng{no.}}
	%\wenbo{More details of the above four nonlinear methods} 
	
	The aforementioned works assume that the Granger causalities remain unchanged between time series channels throughout time~\cite{GRANGER1980329}, which is called ``channel-level Granger causality''.
	In the real world, the time series data is becoming massive, complicated and uncertain, and the causality relations would change along with the sliding windows of the time series data (see Fig.~\ref{fig:dwgc} as an example).
	
	According to the philosopher David Hume, ``cause and effect have temporal precedence\cite{Beauchamp1981Hume}''. This acquiesced that causality or precedence itself is in relation with, and even change dynamically with time.
	In this sense, the constant causality assumption\cite{GRANGER1980329} is not always true and the purpose of this assumption for the traditional Granger causality is to take sufficient long (but definitely not necessary) time series in order to distinguish the random correlation and causalities~\cite{GRANGER1980329}. 
	Some attempts on dynamic causalities have already been done in neural science: the ``dynamic causal modeling'' (DCM) method detects the real-time dynamic causal relationships among neuron clusters in the brain. However, it's by interfering the ``input-state-ouput'' framework in brain and observing its response and cannot directly extract causalities without the model that have relationships like brain~\cite{FRISTON20031273}.

	\begin{figure}[htbp]
		\centering
		\subfigure[Traditional Channel-level Granger causality]{
			\includegraphics[width=5.5cm]{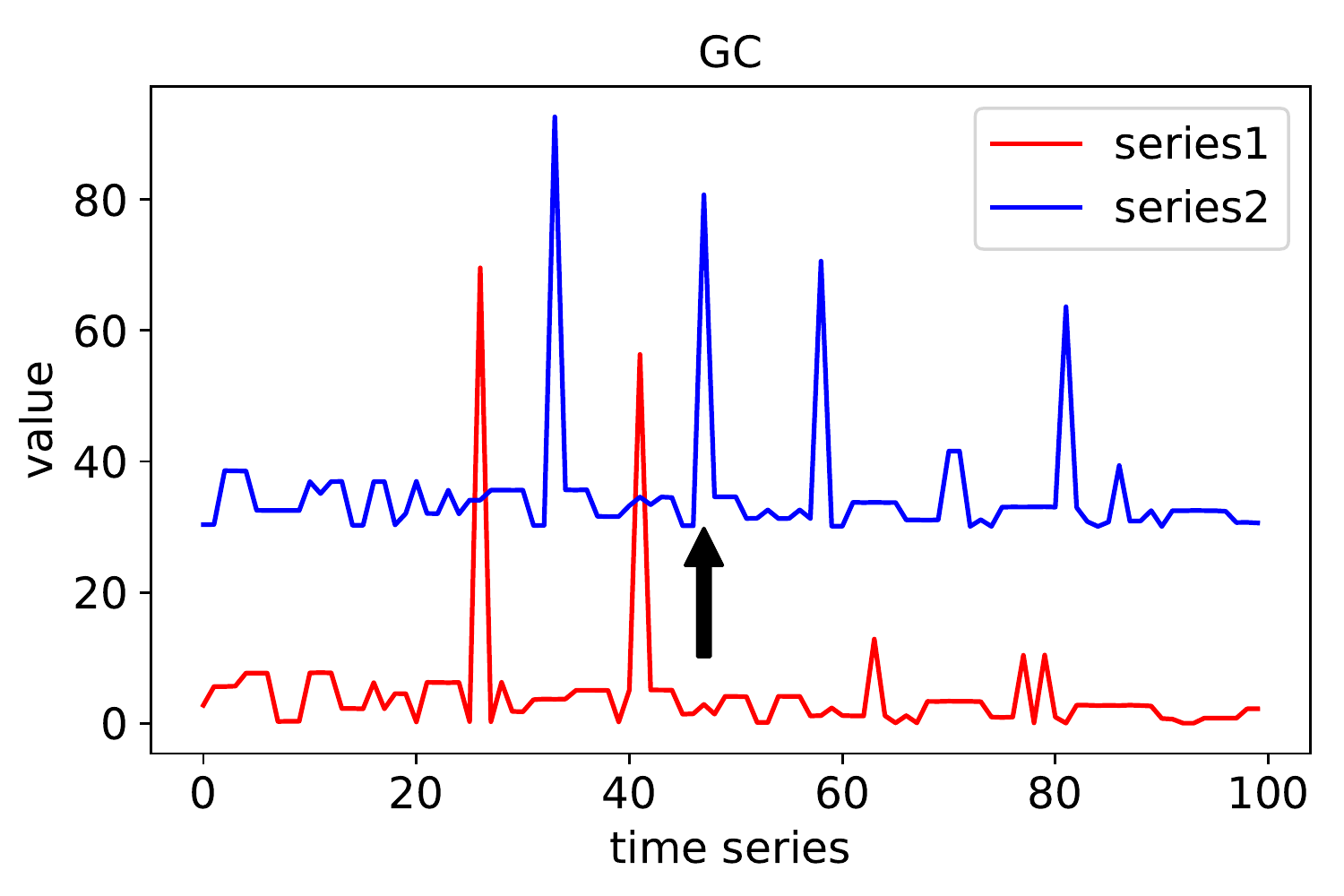}
		}
		\subfigure[Dynamic Window-level Granger Causality]{
			\includegraphics[width=5.5cm]{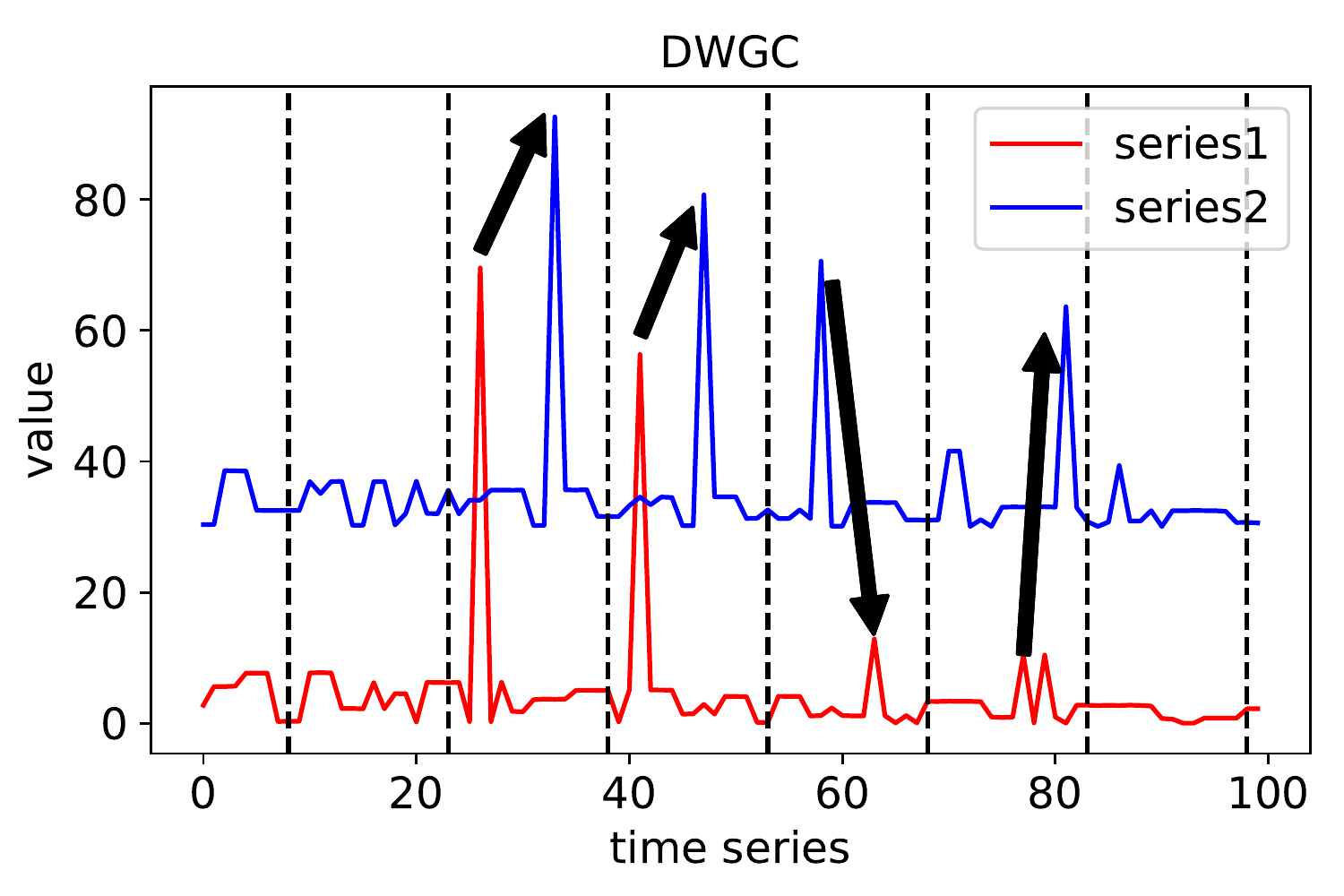}
		}
		\caption{Window-level causality(right) is to analyze dynamic changing causalities on sliding windows on time series. While channel-level method assume causalities remain constant between time series channels(left)}
		\label{fig:dwgc}
	\end{figure}

	\subsection{Our Proposal}
	\noindent In this paper, we present the dynamic window-level Granger causality method~(DWGC) for time series data.
	To capture the dynamic causality relations, we relax the sufficiently long time series assumption and build the causality model on the sliding windows of the time series.
	To capture the non-linearities in the time series forecasting, we use nonlinear autoregressive model to fit the time series and extract the nonlinear features. 
	Based on the NAR model predictions, an F-test is used by comparing the prediction results of NAR models on the sliding windows. 
	%The forecasting errors of the NAR model is used as the assessment criteria of the following F-te
	%which is on the sliding windows level instead of the channel level.
	%Then, for each NAR model forecasting results, we build   calculated F-statistic for each time point according to NAR model fitting results.\\
	Further, to reduce the possible fluctuations of auto-correlation on window-level F-test, we introduce a causality-index matrix and optimize the corresponding causality indexing loss accuracy. 
	We theoretically prove that: 1) the traditional Granger causality is a special case of our dynamic window-level Granger method; 2) the dynamic window-level Granger method outperforms the traditional Granger causality with the causality indexing.
	%it can be improved by the pre-processing mapping vector.
	%\wenbo{use one single sentence to state your theoretical result}
	
	%To address the value instability of the noise and the F-test results on the window level, we introduce a mapping matrix to de-abnormal the anomalous fluctuations to pre-process the raw data, and  In all, We theoretically prove that our methods have two core advantages: \textbf{New task}: our window-level F-test method can find time-dependent causality on window level, whereas the traditional Granger-based method stays at the channel level.\textbf{Higher accuracy:} back to the channel level, the average accuracy of our window-level window F-test is better than that of the traditional Granger method.
	%\wenbo{re-write the above sentence}
	%\wenbo{state your theoretical result here}
	%\wenbo{where is the anomaly detection?}\\

	In the experiments, 
	we implement our DWGC model on three datasets, two synthetic data, and one real-world meteorological data (with prior knowledge). For the synthetic data, we use AR/NAR simulator to generate several synthetic time series. For the real-world meteorological dataset, we examine the obtained dynamic causalities between the El-Nino values and the East Asian monsoon over the seasonal cycle which was already examined in the previous literature.
	
	%the AR simulator experiment splits into two parts: in the first part, we designed a pair of nonlinear time series with different causal relationships on different windows, and mixed the abnormal points with the causal relationship points to evaluate the effect of our model. In the second part, we establish three sequences a,b,c, in which b passes time-dynamic causal information from a to c(or c to a) as an intermediate variable. The first experiment focused on verifying whether our model can accomplish window-level causal detection, and whether abnormal noise and causal effect could be effectively distinguished at the mutation points, while the second experiment focuses on further verifying whether the model can effectively track the dynamic window-level causal direction of intermediate variable. In meteorological case, we examine the causal interaction between el-nino values and the east Asian monsoon over the seasonal cycle, to prove that our model has industrial value and scientific significance in other fields.
	%\wenbo{in this part, you should say how you treat these datasets and do not need to evaluate the results}
	In summary, the contributions of this paper are as follows:
	%\begin{enumerate}
	\begin{itemize}
		%\vspace{-0.8em}
		\item We propose and solve a new task: identify dynamic window-level causality of multi-channel time series.
		\item We theoretically show that the proposed dynamic window-level Granger causality model contains the traditional Granger causality as a special case and is more accurate with the causality indexing.
		\item We conduct the numerical experiments on two synthetic and one real-world datasets and the proposed dynamic window-level Granger method shows that this method can find the accurate window causalities. 
	\end{itemize}
	
	%To the best of our knowledge, this is the first attempt to take account of dynamic and window-level causalities. 
	%Compared with the nonlinear Granger causality method~\cite{tank2018neural}, our method can promote the causal relationship from the channel level to the dynamic window level.
	%Besides, Compared with the abnormal detections in \cite{killick2011efficient}, where authors introduced the PELT method to detect mutation in time series, we give an association between causal reasoning and anomaly detection.  In our method, a de-abnormal mapping vector is used to make the traditional Granger method more effective.
	
	%\wenbo{Make some comparations above}
	
	\section{Related Works}
	\subsection{Preliminaries of Granger Causality}
	\label{sec:related works}
	\begin{comment}
	\wenbo{Point out related works here: 1)causal graph including nonlinear version, 2)deep time series prediction, 3)time series anomaly detection models, 4) list more. This part should be at least one page.}
	\end{comment}
	
	%\subsection{Basics of multi-channel Time Series Causality}
	%In \cite{eichler2013causal}, four different time series causalities are collected and described: structural causality, Granger causality, and sims causality\wenbo{only three?}. Among these four,
	%\wenbo{double check this subsection. make sure the two definitions are compatible}
	\noindent Granger causality is the most widely used causality analysis method for time series data~\cite{granger1969investigating}, which is also the main focus of this paper.
	We use the following sequence to represent the stationary sequence of random variables, i.e, the time series data $i$: 
	\begin{equation}
	{(Y_{i1},Y_{i2},Y_{i3},\cdots,Y_{it},\cdots)},
	\end{equation}
	where ${Y_{i,t}}$ and $Y_{<t}$ are the series data $i$ at and before time $t$. 
	Granger causality defines ${Y_{i}}$ as the cause of ${Y_j}$, if the series $i$ provides useful information when predicting the future values of series $j$:
	\begin{equation}
	E_{t}(g_j(Y_{j,t+1}|Y_{j,<t},Y_{i,<t})) \neq E_{t}(g_j(Y_{j,t+1}|Y_{j,<t})), \label{1}
	\end{equation}
	where $E_{t}$ is the accuracy expectation of the prediction function $g$ on time series. Besides, the premise of \eqref{1} is that channels i and j meet the ``backdoor condition\cite{10.5555/2074158.2074209}'', that is, there is no common interference from confounding factors of other channel. Also, it's worth pointing out that Granger causality is still a type of statistical association as it does not go through the necessary causal identification process of \cite{pearl2018theoretical}.
	
	Traditional Granger causality methods use a linear auto-regressive model for $g$. 
	In this paper, we use the double-headed arrow to express the Granger causality $i$ to $j$:
	\begin{equation}
	Y_{i}{\Longrightarrow}Y_{j}.
	\end{equation}
	Granger causality is considered as a ``precedence'' according to the informal fallacy ``Post hoc ergo propter hoc''. This Latin fallacy means ``after this, therefore because of this''. It shows the causality as the precedence that the follow-up event is caused by the previous event.
	
	The linear Granger causality can be extended to the nonlinear version for better fitting nonlinear sequences. 
	In~\cite{tank2018neural}, a non-linear prediction model, such as the multi-layer perceptron, is defined as $Y_{jt}=g_{j}(Y_{1,<t},\cdots,Y_{d,<t})$.
	Then, we can say $Y_{i}{\Longrightarrow}Y_{j}$, if for all $Y_{i,<t}^{'} \neq Y_{i,<t}$:
	\begin{equation}
	g_{j}(Y_{1,<t},...,Y_{i,<t},...,Y_{d,<t}) \neq g_{j}(Y_{1,<t},...,Y_{i,<t}^{'},...,Y_{d,<t}).
	\end{equation}
	%This method tried to use the parameters of the hidden layer of MLP to represent the coefficients in the AR model to simulate and generalize the Granger causality.
	%Whereas it is difficult to achieve the effect of  ().
	$g$ is represented as MLP(Multilayer Perceptron)/LSTM(Long Short-Term Memory) network in \cite{tank2018neural}. Both linear and nonlinear Granger causality method assumes that the causalities between time series are constant and cannot model the dynamic causalities that lie in the real world.
	%cannot achieve the	dynamic recognition of causality at the window level because the fixed (neuron) parameters are bound to determine fixed causality channel pairs.

	\subsection{Improvements of Granger method to generalize to the nonlinearity/window level}
	\noindent More recently, the deep neural networks are used to get nonlinear Granger causalities~\cite{jones1978nonlinear,chivukula2018discovering,xu2019scalable,duggento2019echo,tank2018neural,he2019causalbg}. In these previous works, the neural network is used to replace the original AR model for sequence fitting and prediction, and then causal reasoning methods (such as counterfactual principle, Granger method) are brought into the framework of neural network. 
	%However, due to the complex node interaction between hidden layers, the impact of input is difficult to quantify explicitly, thus it's difficult to determine whether the results are accurate enough without interference from noise and model errors. 
	%\wenbo{introduce all the works briefly}
	
	On the other hand, Granger's approach has been extended from channel level to window level. Sezen Cekic\cite{cekic2018time} use KL(Kullback-Liebler) divergence instead of F-test to extend the window-level Granger method, first in neuroscience. 
	
	To both deal with nonlinearity and window-level challenges, sliding window method is common, which is also applicable in Granger causality situation. The neuroscience Granger causality extended to the nonstationary case by wavelet transforms or multitapers on sliding window by Mattia F.Pagnotta\cite{PAGNOTTA2018478}, whereas as far as we know, there is no related method that can address both the nonlinear and the window-level problems in a wider application scenario.
	\subsection{Causal Graph for Time Series}
	\noindent Besides the Granger causality, the causal graph is another way to identify the causality between different channels of time series~\cite{xu2019scalable}. The causal graph model can be built based on vector autoregressive Granger analysis(VAR), the generalized additive models(GAMs)~\cite{lutkepohl2005new} or the certain pre-assumed regression models \cite{sommerlade2012inference}.
	However, the causal graph is limited to the structure itself to extend to the dynamic window-level for time series data.
	%Therefore, it's complex to detect the dynamic window-level causal detection using the causal graph.
	
	\subsection{Distinction between causal effects and noise}
	\noindent In each window, it is a subtle topic to distinguish whether the time series trend variations are due to the causal effects or random noise. 
	%Common anomaly detection methods including donut, dagmm, telemanon\cite{article}\cite{bardwell2017bayesian} are not effectively combined with the causal detection algorithm.
	On the one hand, commonly-used time series anomaly detection methods barely consider causal effects when detecting different abnormal noises.
	In \cite{brodersen2015inferring}, although an attempt is made to use anomaly detection to divide the actual sequence into two parts: steady-state structure part and causal influence part, the second part does not make a practical distinction between noise and causality. On the other hand, the traditional Granger causality method keeps the ignorance of noise intervention in most cases. 
	%Moreover, in both of the cases, only the termination point of causality can be detected, instead of the occurrence point.
	\section{Dynamic Window-level Granger Causality}
	\subsection{Naive Dynamic Window-level Granger Causality Model}
	% DWGC without \Phi
	
	\noindent To detect dynamic window-level causality between two data series $Y_i$ and $Y_j$, we first use a sliding window of length $k$ on the same time position $t$ of the two series:
	\begin{equation}\{Y_{i,t},Y_{i,t+1}, \cdots, Y_{i,t+k-1}\},~~~\{Y_{j,t},Y_{j,t+1},\cdots, Y_{j,t+k-1}\}.
	\end{equation}
	
	Aiming at finding the dynamic causality at the window level, we consider two forms of time-series fitting on each sliding window $(t, t+k-1)$: predicting the future values of one series with and without the information from the other series channel, which are similarly used in the traditional Granger causality method. Compared to equation \eqref{1}, we use nonlinear auto-regressive(NAR) model, and use mean square error(MSE)\cite{Aneesh2018Discovering} as $E_t$ to measure two accuracies:
	\begin{eqnarray}
	\label{MSE1}
	L_{1}=E(\mathrm{MSE}({\hat Y^{}_{i,t\sim t+k-1}}, Y_{i,t\sim t+k-1}|Y_{i,<t})),\\ 
	\label{MSE2}
	L_{2}=E(\mathrm{MSE}({\hat Y^{}_{i,t\sim t+k-1}}, Y_{i,t\sim t+k-1}|Y_{i,<t},Y_{j,<t}),
	\end{eqnarray}
	where $\hat Y^{}_{i,t\sim t+k-1}$ means prediction on the sliding window $(t, t+k-1)$. We determined the existence of causality by setting a reasonable  threshold $\epsilon$ based on the value of $F_{statistic}$:
	\begin{equation}
	F_{statistic} = \frac{L_{1}}{L_{2}}. \label{eqn:f-statistic}
	\end{equation} 
	The causality exists between the series $i$ and $j$ on the sliding window $(t, t+k-1)$ if $F_{statistic}$ is larger than the pre-defined threshold  $\epsilon(>1)$.
	
	\subsection{Dynamic Window-level Causality Model}
	% DWGC with \Phi

	%In Sec. \ref{sec:related works}, we defined the decision condition of $Y_{i}{\Longrightarrow} Y_{j}$, here, we improved the causal definition from channel level to window level, k means the window length.
	%\wenbo{this already has a time window, which is totally different from 2.1. explain more explicitly}
	%\begin{equation}
	%Y_{i,t_1} %\stackrel{}{\Longrightarrow}
	%Y_{j,t_2}( i,j \in%{\{1,2,...d\}},t_1< t_2).
	%\end{equation}
	%\wenbo{B or Y???}\wenbo{what do you mean by `d-dimension'?}
	%\wenbo{make sure all formulas are unified and equations are compact}
	\noindent In the naive version of the dynamic window-level causality method, auto-correlation would easily occur with the disturbances of the NAR model along the time series.
	To address this problem, we introduce the causality indexing matrix $\Phi$ in our method and then decompose the causal effects and auto-regressive correlations. The original time series includes two factors: the cross-correlation and the auto-correlation. By converting to windows-level, the auto-correlation becomes more unstable and conceals the causal effects in the cross-correlation. 
	The auto-regression correlations in one single series is represented in $L_1$ in Eqn.~\eqref{MSE1}, intuitively, when auto-correlation is unusually large on a local window, it will make the window-level f-test result higher than the normal value, which may affect the overall accuracy of the model, or affect the overall recall rate conversely. We use a scale function $h$ to scale down the auto-correlation and adopt a corresponding causal indexing matrix $\Phi$, to measure the likelihood that each time point will serve as a starting or ending point for causality. The indexing loss as 
	\begin{equation}
	Loss =\sum_{m,i}\mathrm{KL}\left(\Phi_{m}^{i},\{{}{h(\hat Y^{}_{i,q}- Y_{i,q})^2}\}\right), q=m, m+1,...m+k-1,
	\label{eqn:model_loss}
	\end{equation}
	where $\hat Y^{},Y$ are the prediction results and the real series respectively, $i$ is the channel index and $m$ is the starting point of each sliding window with length $k$.
	By optimizing this loss, $\Phi$ can be used to scale down the original series data $Y$ with large auto-correlations as:
	\begin{equation}
	\begin{aligned}
	\begin{bmatrix}A_{i,t}\\A_{i,t+1}^{}\\...\\A_{i,t+k-1 }\end{bmatrix} = 
	%\begin{bmatrix}&\phi^{i}_{t},\phi^{i}_{t+1},...\phi^{i}_{t+k-1} \end{bmatrix} *
	\Phi^{i}_{t\sim t+k-1} *
	\begin{bmatrix}Y_{i,t}\\Y_{i,t+1}\\...\\Y_{i,t+k-1}\end{bmatrix}, i=1,2,...d,
	\end{aligned}
	\label{eqn:index_series}
	\end{equation}
	where $*$ is Hadamard product.
	With the causality indexing $\Phi$, we scale down the auto-correlation and get the reweighted series $A$.
	In this paper, we use the following scaling function:
	\begin{equation}
	h=(\alpha-\tanh(\cdot)),\alpha>1.
	\label{eqn:scale function}
	\end{equation}
	It is worth mentioning that the selection of the scaling function $h$ can be further improved by certain regularization item and this will be analyzed in the Appendix~\ref{proof-theorem-2}.
	%\zhiheng{Section 4 ends with an additional sufficient condition, not a necessary condition, so I made a little revision of this sentence.}
	
	\begin{algorithm}[ht]
		\SetAlgoLined
		\KwData{Multi-channel time series $Y_{t}$ and predefined F-test threshold $\epsilon$.}
		initialization, $\Phi$ set to all-one\;
		\While{NAR forecast loss and causal indexing loss converge}{
			Reweight original time series using causal index matrix $\Phi$ via Eqn.~\eqref{eqn:index_series}\;
			Train an NAR model using reweighted series $A_{t}$ as the input\;
			Finding dynamic causalities via Eqn.~\eqref{eqn:f-statistic-reweight}\;
			In all the window pairs with the detected causalities, optimize Loss \eqref{eqn:model_loss} to optimize $\Phi$\;
		}
		%Exact causality and the starting points via Eqn.~\eqref{eqn:many-to-many-causality}\;
		Position causality from window-level to point-to-point level by \eqref{eqn:p-to-p-causality}.
		
		\KwResult{Dynamic window level causalities}
		
		\caption{ Framework of our method(DWGC)}
		\label{alg:Framwork}  
	\end{algorithm}
	
	We test the dynamic window-level causality using the causality indexing as:
	\begin{equation}
	F_{statistic} = \frac{L_{1}}{L_{2}}
	=\frac{E(MSE({\hat A^{}_{i,t\sim t+k-1}}, A_{i,t\sim t+k-1}|A_{i,<t}))}{E(MSE({\hat A^{}_{i,t\sim t+k-1}}, A_{i,t\sim t+k-1}|A_{i,<t},A_{j,<t}))}.
	\label{eqn:f-statistic-reweight}
	\end{equation} 
	
	We further obtain the starting and ending points in the sliding window by finding the maximum of the index matrix:
	\begin{eqnarray}
	Y_{i,t_1} &\stackrel{}{\Longrightarrow} & Y_{j,t_2}\\
	\{t_1, t_2\}=\mathop{\max}&&\{\Phi_{t_1}^{i}+\Phi_{t_2}^j|t_1<t_2,~~~t_1,t_2 \in (t,t+k-1)\}.
	\label{eqn:p-to-p-causality}
	\end{eqnarray}

	In our dynamic window-level Granger causality (DWGC) method, we alternatively optimize the NAR forecasting model and the causality index loss Eqn.~\eqref{eqn:model_loss} and extract causalities after the two loss function converge.
	The final procedure is outlined in Algorithm~\ref{alg:Framwork}. 
	%In each prediction process, we use the information of fitting error to guide $\Phi$ until model converges.
	
	\section{Theoretical Analysis of DWGC}
	\noindent In this section, we give the theoretical analysis of the proposed dynamic window-level Granger causality method(DWGC).
	We first give the formulation preparation and then prove that 1) DWGC without causality indexing is a special case of the traditional Granger causality method and 2) DWGC with causality is more accurate than the traditional Granger causality method for those causality pairs.
	
	We consider the sample pairs whose expectation of $F_{statistic}$ $F_0$ is larger than the predefined threshold $\epsilon$. These samples are expected to be tested as causal in the traditional Granger causality method.
	
	\subsection{Formulation Preparation}
	\noindent The time series observations $Y_t$ can be decomposed into two parts: real data and Gaussian noise. For simpilicity, We use the standard Gaussian for the random noise and get the following decomposition of the time series $i$: 
	\begin{equation}
	Y_i=Y_i^{\text{real}}+\gamma_i, \gamma \sim \mathcal{N}(0,1)
	\end{equation}
	We denote the model predictions of the NAR model with/without another causality source series channel $j$ as $\hat{Y}_{i}$ and $\hat{Y}_{i|j}$.
	
	Then the F-statistic of Eqn.~\eqref{eqn:f-statistic}, i.e, the ratio between the MSE of the NAR models with/without the channel $j$, can be turned to
	\begin{equation}
	\begin{aligned}
	F_{statistic} = &\frac{L_1}{L_2} = \frac{{\sum_{t}{(\hat{Y}_{it}-Y^{\text{real}}_{it}-\gamma)}^2}}{{\sum_{t}{(\hat{Y}_{it|j}-Y^{\text{real}}_{it}-\gamma)}^2}}\\
	=& \frac{\sum_{t}{(\hat{Y}_{it}-Y^{\text{real}}_{it})}^2+\sum_{t}^{} \gamma_{it}^2 -2\sum_{t} \gamma_{it}(\hat{Y}_{it}-Y^{\text{real}}_{it})}{\sum_{t}^{}{(\hat{Y}_{it|j}-Y_{it}^{\text{real}})}^2+\sum_{t}^{} \gamma_{it}^2 -2\sum_{t}^{} \gamma_{it}(\hat{Y}_{it|j}-Y_{it}^{\text{real}})}\\
	\approx& \frac{\sum_{t}{(\hat{Y}_{it}-Y^{\text{real}}_{it})}^2+\sum_{t}^{} \gamma_{it}^2 }{\sum_{t}^{}{(\hat{Y}_{it|j}-Y_{it}^{\text{real}})}^2+\sum_{t}^{} \gamma_{it}^2}.
	\label{eqn:f-statistic-window}
	\end{aligned}
	\end{equation}
	The square sum of the standard normal noise $\sum_{t}\gamma_{t}^{2}$ follows the chi-square distribution $\chi^2(k)$, where $k$ is the length of the sliding window. When the window length $k$ is reasonably large, the third term in the above formulation can be omitted. The details of the omitting derivation can be seen in Appendix \eqref{eq:lemma 1}.
	%\wenbo{double check and complete}
	
	%After the analytical expression of $F_{statistic}$, we begin to design an evaluation index to the effect of Granger causal detection. We first compute the probability density function $f_{F_{statistic}}(\epsilon)$ 
	%\wenbo{what do this equation mean??}. 
	%\wenbo{What is m}. 
	%After that, we compare the performance of the window/channel level Granger approach by detecting the value $\int_{1}^{+\infty}f_{F_{statistic}}(\epsilon)d\epsilon$ with respect to k(the traditional Granger method is the case when k is the whole channel length).
	%That is to say, on the basis that there are a double-channel causal sequences with priori causality(if without noise, $F_{statistic} = \epsilon_0>1$), we detect the change of the probability of $F_{statistic}$ falling within a reasonable threshold range on each window length k.

	\subsection{DWGC without causality indexing degenerates to Traditional Granger causality}
	%Intuitively, when the sliding window grow longer, our naive dynamic window-level Granger causality (DWGC) method without causality indexing degenerates to the traditional Granger causality method. 
	\noindent We analyze the distribution of the F-statistic of the naive DWGC and give the following result.
	\begin{theorem}
		\label{theorem:dwgc1}
		For the time series sliding windows with causalities, the probability of F-statistic of naive DWGC larger than the threshold $P(F(k)>1)$ is a monotone increasing function for $k$ in the case where $k$ is sufficiently large. 
	\end{theorem}
	
	We prove the above theorem using the series expansion of $P(F(k)=\epsilon)$ on each $k$, and leave the details in Appendix~\ref{proof of theorem 1}.
	
	This theorem can be intuitively analyzed from the expression of the F-statistic in Eqn.\eqref{eqn:f-statistic-window}: 
	for numerator, $E(\sum_{t}{(\hat{Y}_{it}-Y^{\text{real}}_t)}^2+\sum_t^k \gamma_t^2)=k(E(\hat{Y}_{it}-Y^{\text{real}}_t)^2+1)$, for denominator, $E(\sum_{t}{(\hat{Y}_{it|j}-Y^{\text{real}}_t)}^2+\sum_t^k \gamma_t^2)=k(E(\hat{Y}_{it|j}-Y^{\text{real}}_t)^2+1)$, so when k is sufficiently large, Eqn.\eqref{eqn:f-statistic-window} can be approximate to $\frac{(\hat{Y}_{it}-Y^{\text{real}}_t)^2+1}{(\hat{Y}_{it
			|j}-Y^{\text{real}}_t)^2+1}(>1)$.
	%\wenbo{this paragraph is false, rewrite}\zhiheng{revised, is it ok}
	\begin{comment}
	The variance of the numerator and the denominator in the F-statistic would be smaller with the increase of the window length $k$. For example $\frac{\sum_{t}^{k}\gamma_t^2}{k}$. The variance of $\frac{\sum_{t}^{k}\gamma_t^2}{k}$ gets smaller with the increase of k ($\sigma(\frac{\sum_{t}^{k}\gamma_t^2}{k})=\frac{2}{k})$, so do other items. Therefore, as window length k increases, the $F_{statistic}$ steadily approaches a limit value approximately equal to $\epsilon(>1)$, so the probability $P(F(k)>1)\approx 1$ when k is large enough, while $P(F(k)>1)
	< 1$ when k is small.\wenbo{this paragraph is false, rewrite}
	\end{comment}

	Without the causal indexing, the larger the sliding window length $k$ is, the more accurate the DWGC method without the causality indexing would be. When $k$ goes to infinity, the DWGC method without the causality indexing degenerate to the traditional causality method.
	
	Fig.~\ref{fig:f-statistic}(left) is the scatter heatmap of F-statistic of the naive DWGC method. The x-axis is the sliding window length and the y-axis is the is the value of the F-statistic. This figure is done based on an experiment for a synthetic AR-simulation data with causalities, which will be further illustrated in the experiment section. As can be seen from the results, the expectation of F-statistic of naive DWGC is monotone increasing for the window length $k$.

	\begin{figure}
		\centering
		\begin{minipage}[c]{0.5\textwidth}
			\centering
			\includegraphics[height=4.5cm,width=7cm]{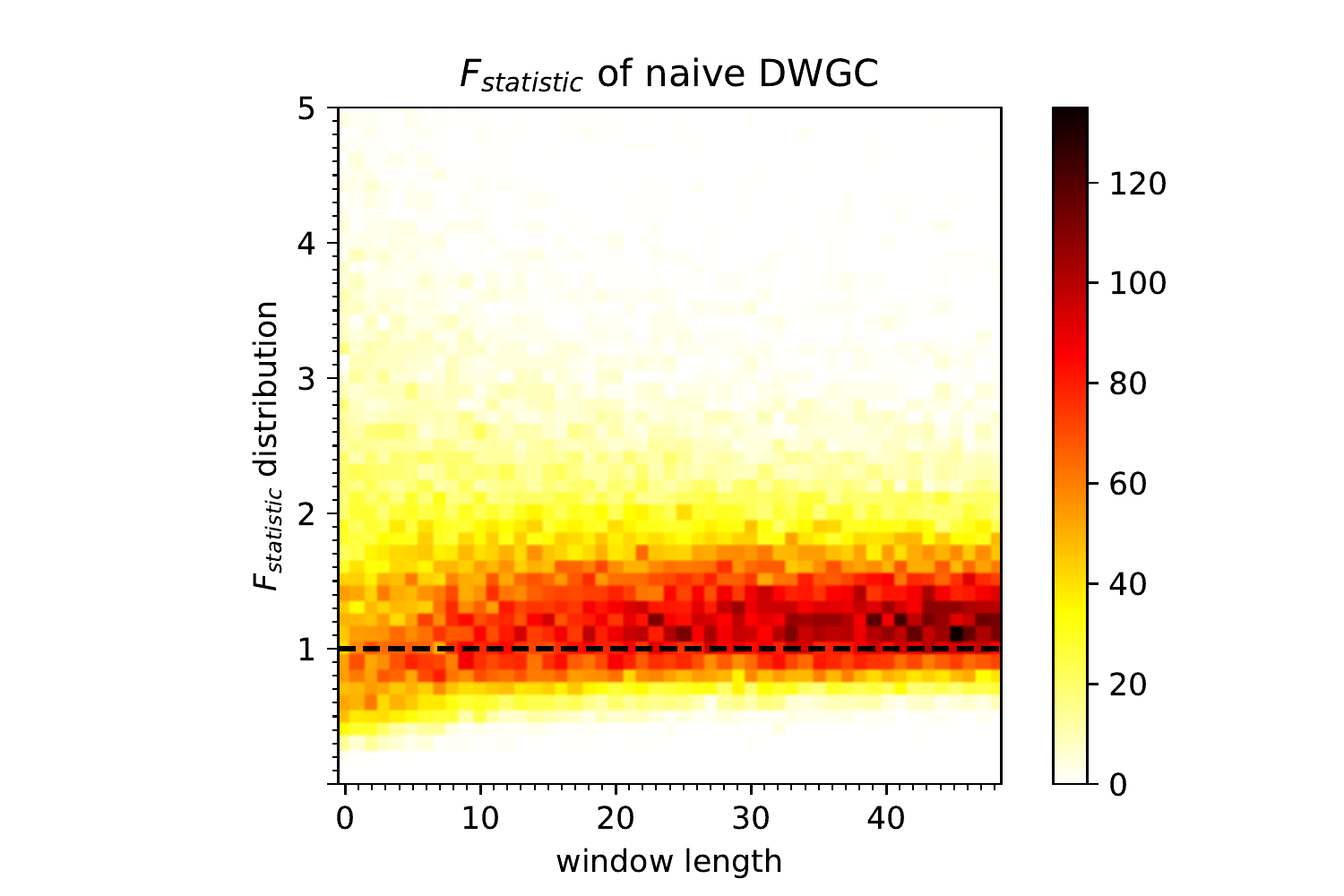}
		\end{minipage}%
		\begin{minipage}[c]{0.5\textwidth}
			\centering
			\includegraphics[height=4.5cm,width=7
			cm]{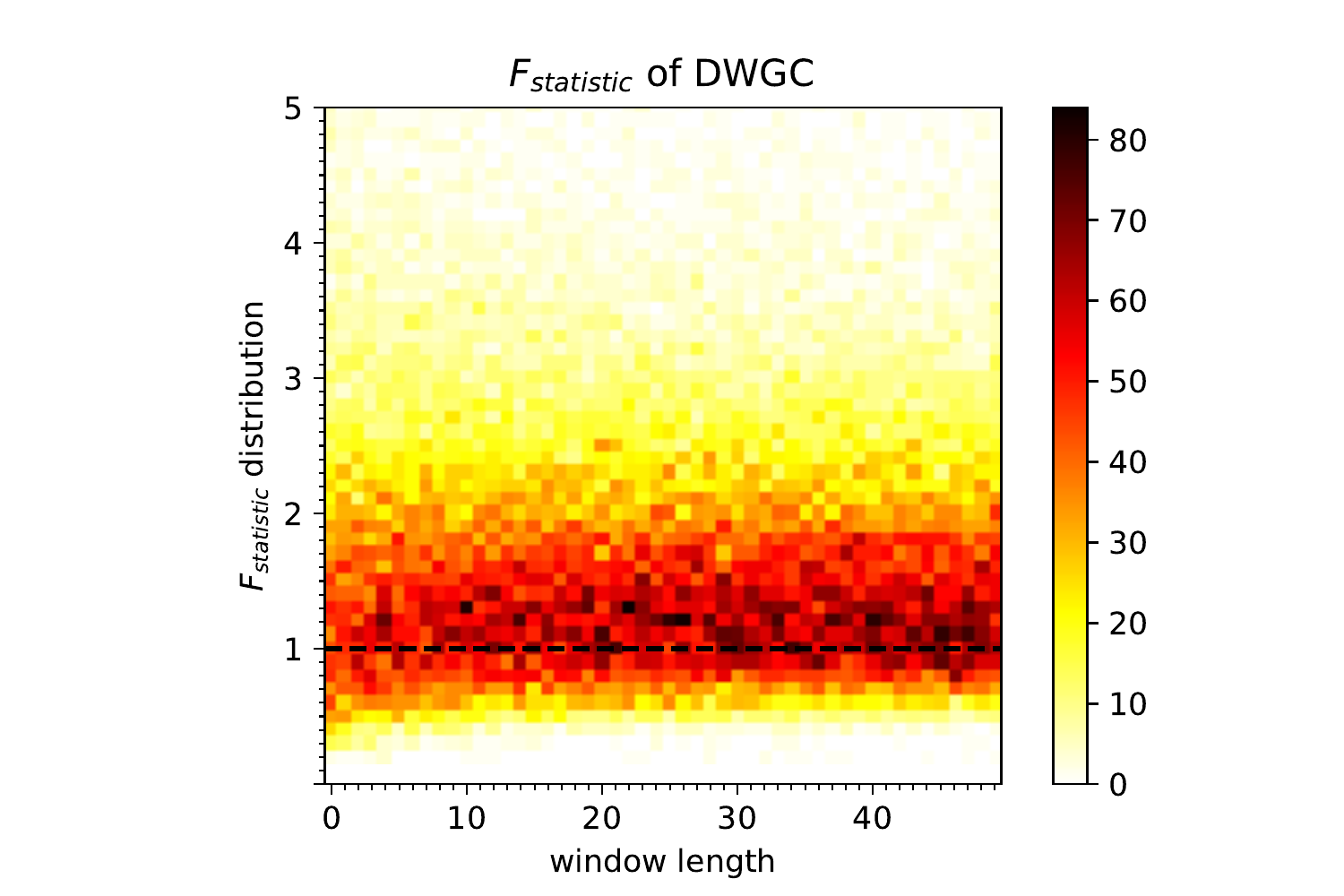}
		\end{minipage}
		\caption{The scatter heatmap of  $F_{statistic}$ results from the multiple tests. The tested data are the synthetic AR simulation data pairs with causalities. The causality analysis methods are naive DWGC without causal indexing (left) and DWGC with causal indexing (right). The black dashed line is the causality threshold $\epsilon=1$. The area above the black dashed line is causal and the area below is not.}
		\label{fig:f-statistic}
	\end{figure}

	\subsection{DWGC with causality indexing generate more accurate causality result}
	\noindent In this section, we give the sufficient condition of $\Phi$ to improve our DWGC method.
	\begin{theorem}
		\label{theorem:dwgc2}
		Certain causality-indexing $\Phi$ exists to improve the accuracy of our DWGC causality result on each window length k.
	\end{theorem}
	
	We prove the above theorem by adopting the series expansion of $F_{statistic}$ in \eqref{theorem:dwgc1}, and take the increasing of each series after adding $\Phi$ as a sufficient condition to improve our DWGC effect.
	
	Existence of $\Phi$ can be intuitively illustrated as follows. The important reason why the Eqn.~\eqref{eqn:f-statistic-window} is unstable before adding $\Phi$ is that $(\hat{Y}_{it}-Y^{\text{real}}_t)$,$\gamma_t$ and $(\hat{Y}_{it|j}-Y^{\text{real}})$ are all independently distributed Gaussian variables, so it gives a unstable influence on $F_{statistic}$. However, with the causal indexing $\Phi$, an causal reweighting method is to assign specific weights to establish correlation between each item by $\Phi$. For example, when $(\hat{Y}_{it}-Y^{\text{real}}_t)$ is observed to be significantly out of the normal range, we give $\Phi$ to both scale down the fitting error and noise, so as to offset the negative effect of abnormal fitting values on time point t on the whole $F_{statistic}$.
	
	Besides, Eqn. \eqref{eqn:model_loss} give another view to explain the existence of $\Phi$. During the experiment, we can include this theoretical sufficient condition into the regularization term of loss function\eqref{eqn:model_loss} to help iterative optimization. In a word, our DWGC method's effect can be improved by $\Phi$ satisfying certain condition. Concrete form of this sufficient condition is shown in Appendix~\ref{proof-theorem-2}.
	
	\section{EXPERIMENTS}
	\noindent In this section, we present the empirical result comparisons of the dynamic window-level causality method.
	%We share our code and data publicly anonymously in Appendix.~\ref{sec:code}.

	\begin{figure}
		\centering
		\begin{minipage}[c]{0.5\textwidth}
			\centering
			\includegraphics[height=4.5cm,width=6cm]{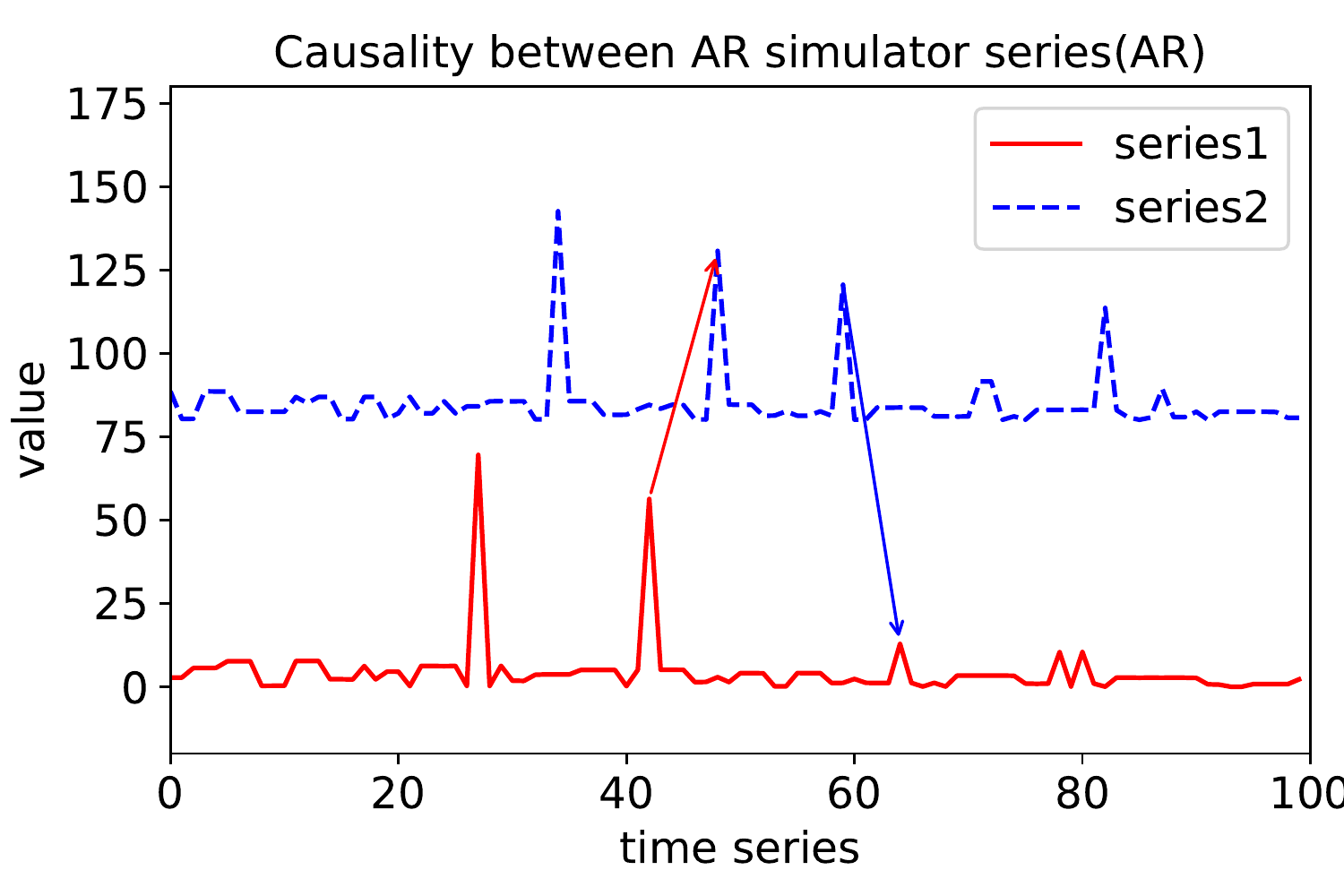}
		\end{minipage}%
		\begin{minipage}[c]{0.5\textwidth}
			\centering
			\includegraphics[height=4.5cm,width=6
			cm]{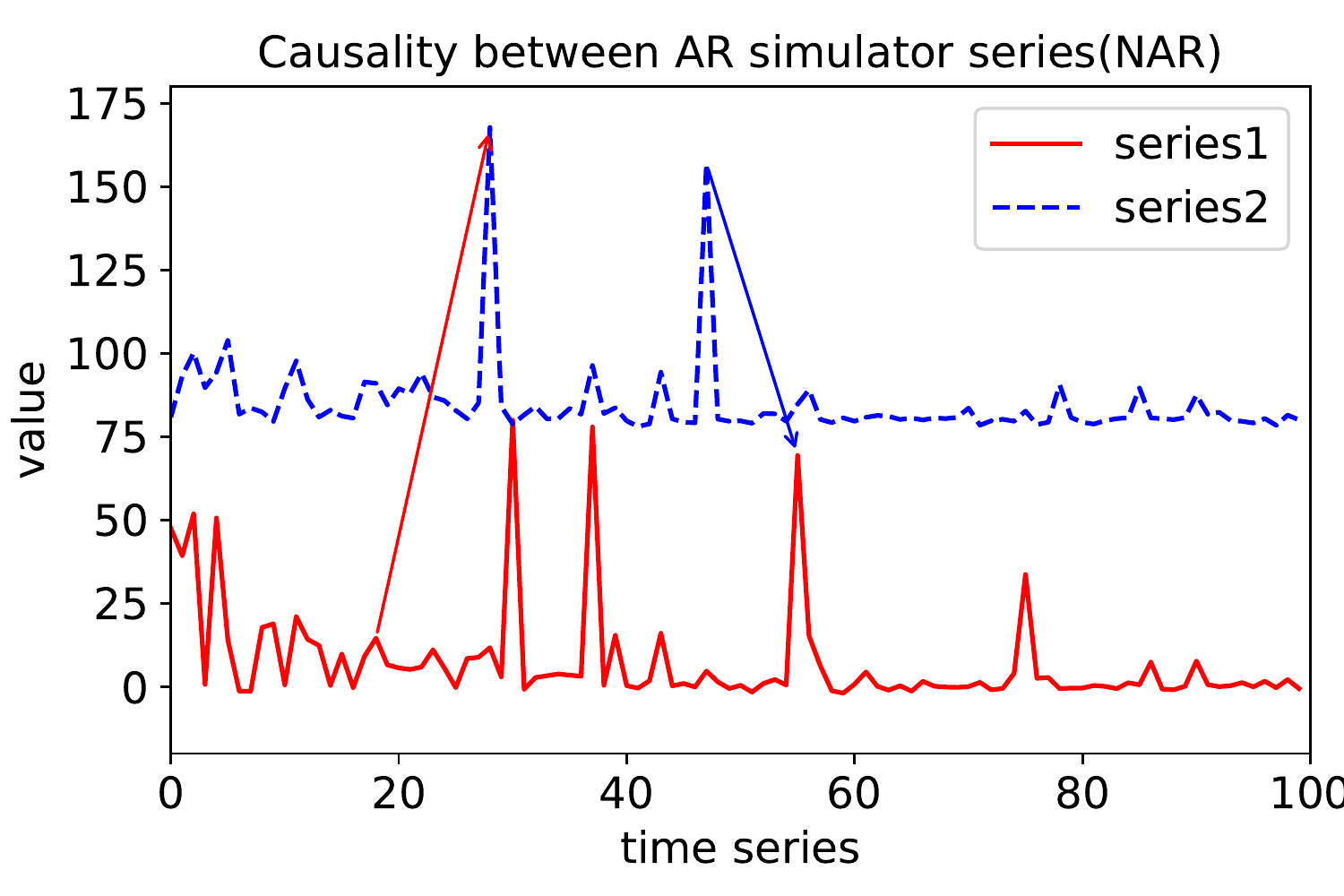}
		\end{minipage}
		\caption{The partial simulation data in AR/NAR experiment, the blue/red curve represents two channels of the time series, arrows represent causal relationships. The starting position of each arrow is at the point when m1(m2) takes an abnormally large value. In this partial display, some of mutations are caused by causal effects and some are abnormal noise. It is easy to confuse causal effects with noise by direct observation.}
		\label{fig:ar-data}
	\end{figure}

	\subsection{Results on Synthetic AR/NAR Simulation Dataset}
	\noindent We first construct the dataset using AR and NAR simulations. The construction details are as follows:
	1) The linear AR simulation construction:
	We simulate two linear AR time series with a random lag value randomly picked from one to nine. The initial value of $T_1,T_2$ is $ T_{i,t}=\left\{
	\begin{aligned}
	t & ~ & \ t \in [1,5] \\
	11-t & ~ & \ t \in [6,10]
	\end{aligned}
	\right. ,$ i=1,2.
	\begin{comment}
	$T_{i,t}$\begin{cases}
	t& \text{t \in (1,5)}\\
	10-t& \text{t \in (1,10)}
	\end{cases} 
	\end{comment}
	\begin{equation}
	\begin{aligned}
	%T_{i} = A_1 * causal\_ effect * T_{i-count} +\epsilon_{noise}
	&T_{1,t} = m_1 T_{2,t-lag}+0.02\mathcal{N}(0,1)\\
	&T_{2,t} = m_2 T_{1,t-lag}+0.02\mathcal{N}(0,1),\\
	&m_1,~m_2 = \begin{cases}
	0.9& \text{p=0.95}\\
	10& \text{p=0.05}
	\end{cases}
	\end{aligned}
	\end{equation}

	2) The non-linear AR simulation construction: we simulate two non-linear AR time series with a random lag value randomly picked from one to nine. The initial value of $T_1,T_2$ is $ T_{i,t}=\left\{
	\begin{aligned}
	t & ~ & \ t \in [1,5] \\
	11-t & ~ & \ t \in [6,10]
	\end{aligned}
	\right. ,$ i=1,2.
	\begin{equation}
	\begin{aligned}
	%T_{i} = A_1 * causal\_ effect * T_{i-count} +\epsilon_{noise}
	&T_{1,t} = m_1 \mathrm{Re}(\sqrt{{T_{2,t-lag}}^2-1})+N(0,1)\\
	&T_{2,t} = m_2 \mathrm{Re}(\sqrt{{T_{1,t-lag}}^2-1})+N(0,1)\\
	&T_{3,t} = \sin(0.1t),
	\end{aligned}
	\end{equation}
	where $ m_1=\left\{
	\begin{aligned}
	10 & ~ & \ T_{3,t}>0.9 \\
	0.9 & ~ & \ T_{3,t}<0.9
	\end{aligned}
	\right. ,
	m_2=\left\{
	\begin{aligned}
	10 & ~ & \ T_{3,t}<-0.9
	\\
	0.9 & ~ & \ T_{3,t}>-0.9
	\end{aligned}
	\right. $,~$lag \in \{1,2,3,...9\}$, and $\mathrm{Re(\cdot)}$ takes the real part of the square root.
	%\wenbo{insert details of NAR data generation here}
	\subsubsection{Experimental Results}
	\noindent A fragment of the AR/NAR simulation data is shown in Fig.~\ref{fig:ar-data} (left: AR and right: NAR). The arrows means the lag relation between the two series and we consider them as the ground truth of the dynamic causalities.
	
	For the two AR/NAR simulation datasets, we first pre-process the data to make sure that the time series are stationary. The reweight function $g$ for causaling index (Eqn.~\eqref{eqn:scale function}) as $(\frac{6}{5}-\tanh(\cdot))$, $\epsilon = 1$. Step length is taken as the window length.

	\begin{table}[ht]
		\caption{Causality Recall of naive DWGC and DWGC(ours) on two AR/NAR simulation datasets}
		\centering
		\begin{tabular}{l|c|ccccc|}
			\toprule
			dataset             &\diagbox  {method}{window length}   & 10          & 20         & 30         & 100        \\ \midrule
			\multirow{2}{*}{AR simulations}  & naive DWGC & 0.48(0.06) & 0.49(0.05) & 0.77(0.03) & 0.77(0.02) \\ 
			& DWGC(ours)       & 0.72(0.03) & 0.72(0.05) & 0.80(0.03) & 0.88(0.05) \\ 
			\hline
			\multirow{2}{*}{NAR simulations} & naive DWGC & 0.59(0.06) & 0.60(0.06) & 0.73(0.06) & 0.87(0.03) \\
			& DWGC(ours)       & 0.84(0.05) & 0.90(0.02) & 0.84(0.08) & 0.88(0.02) \\ 
		\end{tabular}
		\label{table:ar-result}
	\end{table}
	
	\begin{table}[ht]
		\caption{Causality Accuracy of naive DWGC and DWGC(ours) on two AR/NAR simulation datasets}
		\centering
		\begin{tabular}{l|c|ccccc|}
			\toprule
			dataset &\diagbox  {method}{window length}   & 10          & 20         & 30         & 100        \\ \midrule
			\multirow{2}{*}{AR simulations}  & naive DWGC & 0.22(0.06) & 0.43(0.05) & 0.82(0.03) & 0.90(0.02) \\ 
			& DWGC(ours)   &0.23(0.03) & 0.45(0.05) & 0.84(0.03) & 0.92(0.05) \\ 
			\hline
			\multirow{2}{*}{NAR simulations} & naive DWGC & 0.42(0.07) & 0.76(0.09) & 0.93(0.06) & 1.00(0.04) \\
			& DWGC(ours)       & 0.44(0.05) & 0.79(0.04) & 0.94(0.05) & 1.00(0.03) \\ 
		\end{tabular}
		\label{table:ar-result2}
	\end{table}

	On each window, if we detect $F_{statsitic}>1$ with at least a set of causality pairs in this window, our causality extraction on this window is successful. In this case, we calculate the accuracy/recall of the naive DWGC and DWGC methods on the data points with causalities in Table \ref{table:ar-result}. As can be seen from the results, for the recall/accuracy rate, DWGC method performs better than naive DWGC method when the sliding window length is generally small. Besides, The recall rate of naive DWGC method increases with the sliding window length, which certifies the theoretical results.
	%Through AR/NAR experiment, we verify our model can track the direction change of a specific dynamic non-linear causal relationship.
	
	%\subsection{The seasonal cyclical causal effect of the Asian monsoon and El Niño(ENSO)}
	\subsection{Results on Climate Dataset for ElNiño}
	\noindent In this part, we verify the DWGC model on a real climate dataset, which has widely recognized seasonal causalities.
	
	\subsubsection{Academia Knowledge of ElNiño and monsoon}
	\noindent The climate academia already have the following definition on ElNiño: a climate phenomenon in the Pacific equatorial belt where the ocean and atmosphere interact with each other and lose their balance~\cite{ramage151971}. While monsoon generally refers to the seasonal conversion of atmospheric circulation and precipitation in tropical and subtropical areas, and two parameters are used to measure its strength: OLR(Outgoing Longwave Radiation) and MKE(Monsoon Kinetic Energy). 
	
	The causal interaction between ENSO and the east Asian monsoon has been extensively explored:
	\begin{enumerate}
		\item The causality $\text{ENSO} \Longrightarrow \text{MKE/OLR}$ and $\text{ENSO} \Longrightarrow \text{MKE/OLR}$  exists, especially in autumn and winter.
		\item The reverse causal effect $\text{MKE/OLR} \Longrightarrow \text{ENSO}$ and $\text{MKE/OLR} \Longrightarrow \text{ENSO}$ also exists in spring and summer, the strength change of the Asian monsoon will in turn trigger the formation of ENSO event. 
	\end{enumerate}
	
	However, in recent decades, the accuracy of this ENSO-based causal model seems to have a downward trend.  Therefore, the DWGC method would be helpful for analyzing the data with dynamic causalities like this.
	%\wenbo{this part should be rewritten to be more precise. better list the causalities from the  academia knowledge .}

	\begin{figure}
		\centering
		\begin{minipage}[c]{0.5\textwidth}
			\centering
			\includegraphics[height=4.5cm,width=6cm]{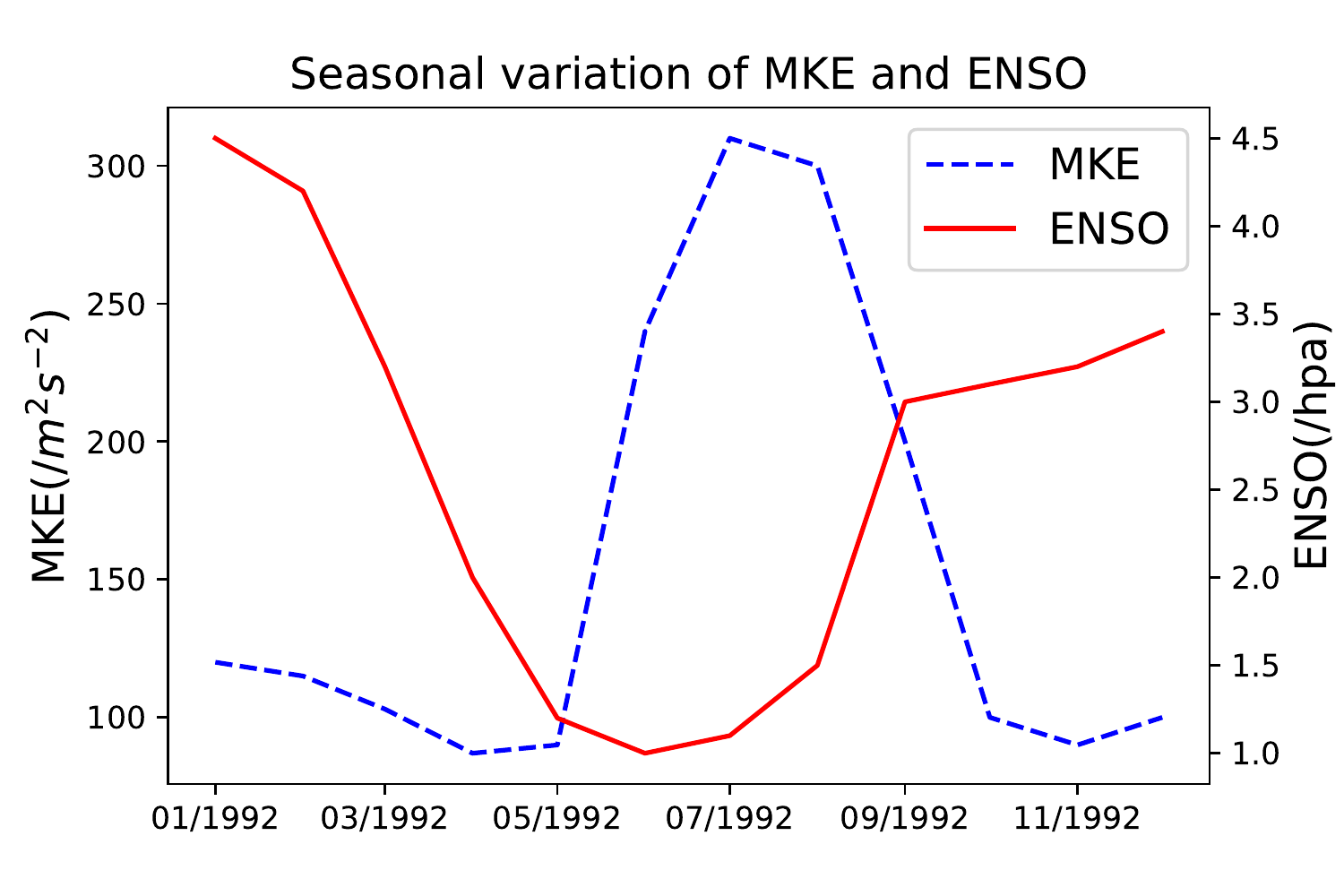}
		\end{minipage}%
		\begin{minipage}[c]{0.5\textwidth}
			\centering
			\includegraphics[height=4.5cm,width=6
			cm]{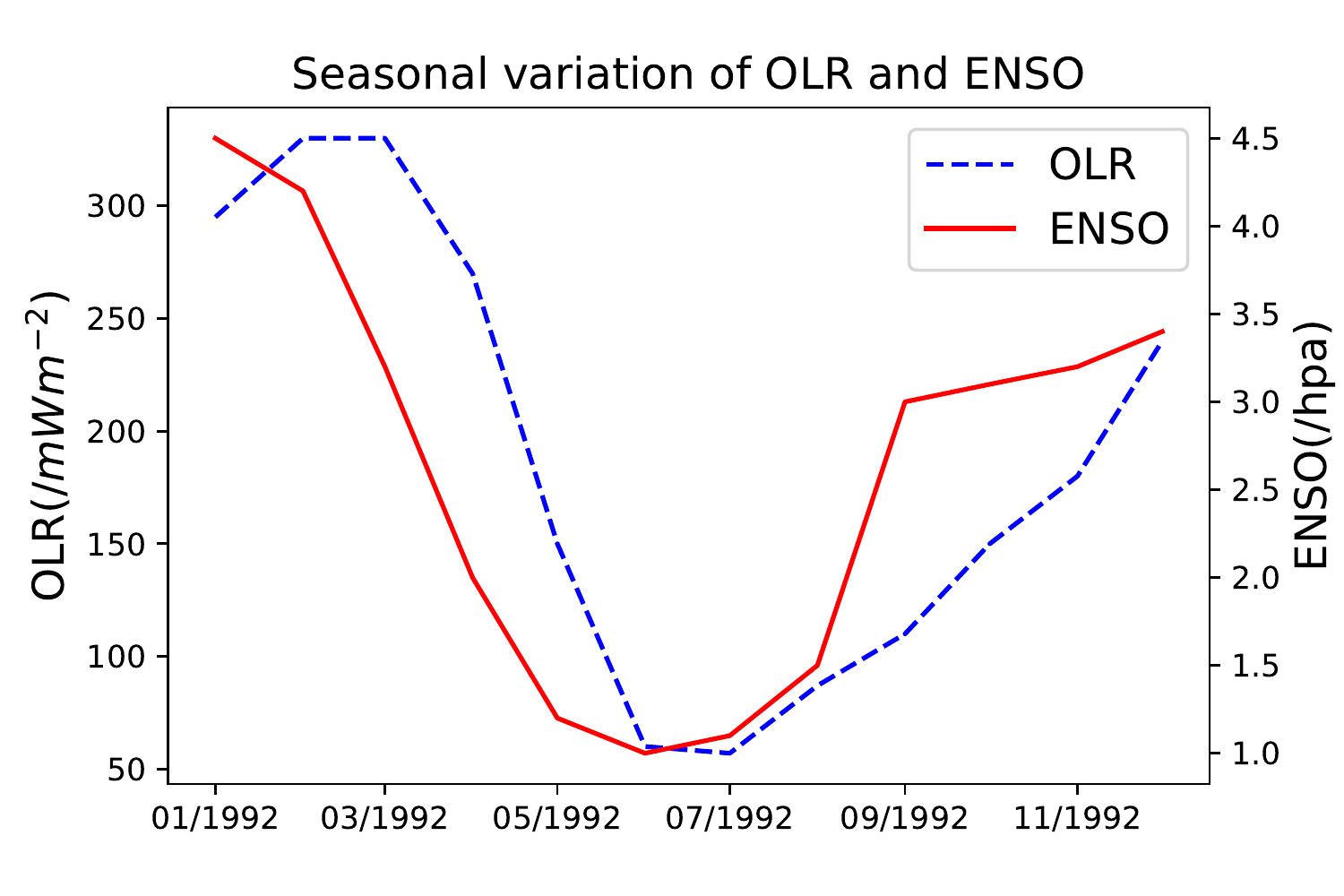}
		\end{minipage}
		\caption{Original data of MKE/OLR and ENSO.}
		\label{fig:raw-data-enso}
	\end{figure}

	\subsubsection{Experiment} 
	\noindent We used ENSO and Asian monsoon data in 1992 and the series trends can be seen from the part of the raw data in Fig.~\ref{fig:raw-data-enso}. We take the first four months as training data and the rest as the testing data~\cite{Yang2018Selective}. For the comparison with prior knowledge, we selected the window length as one month.
	%\wenbo{add the data link if any}. 
	
	\begin{figure}
		\centering
		\begin{minipage}[c]{0.5\textwidth}
			\centering
			\includegraphics[height=4.5cm,width=6cm]{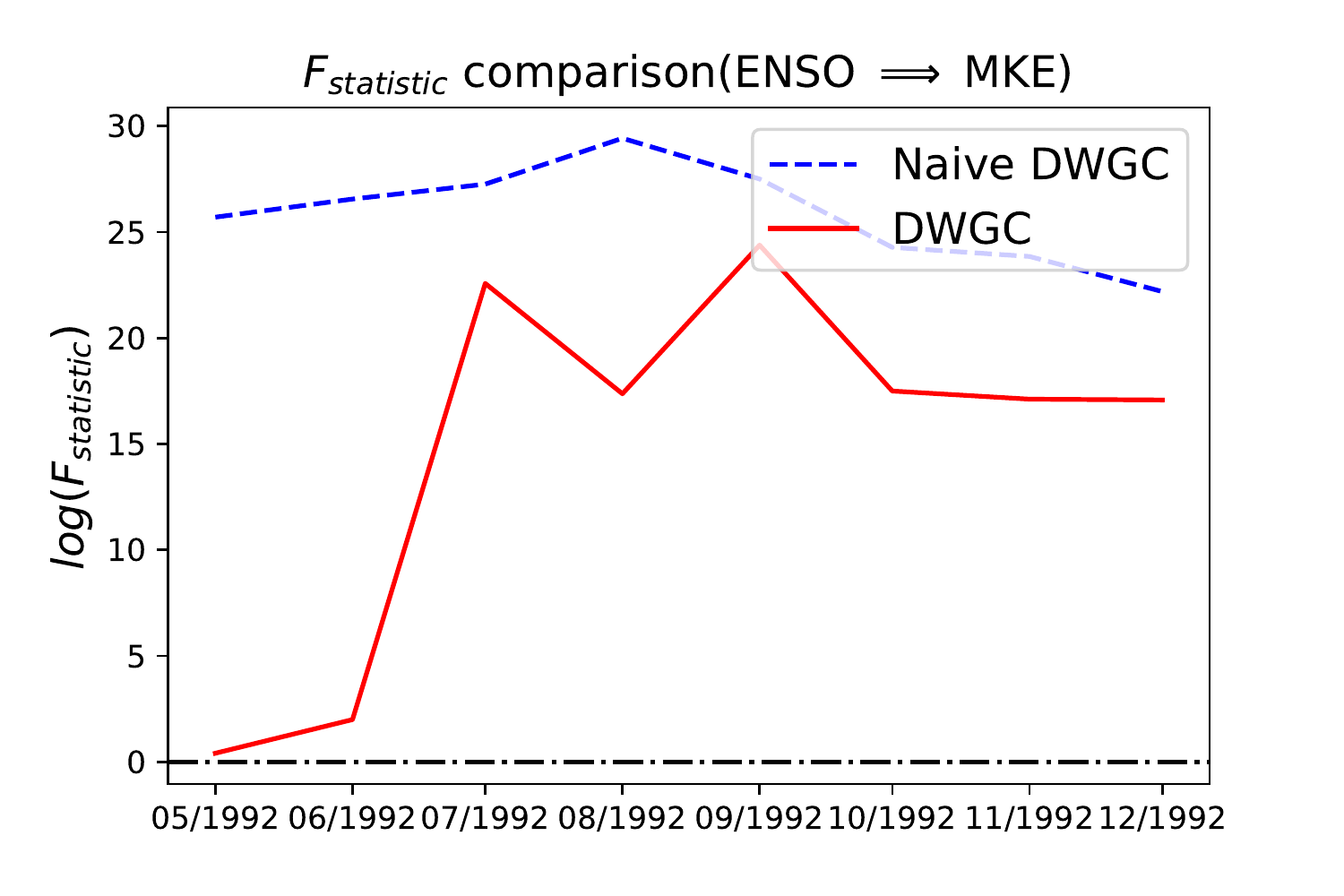}
		\end{minipage}%
		\begin{minipage}[c]{0.5\textwidth}
			\centering
			\includegraphics[height=4.5cm,width=6
			cm]{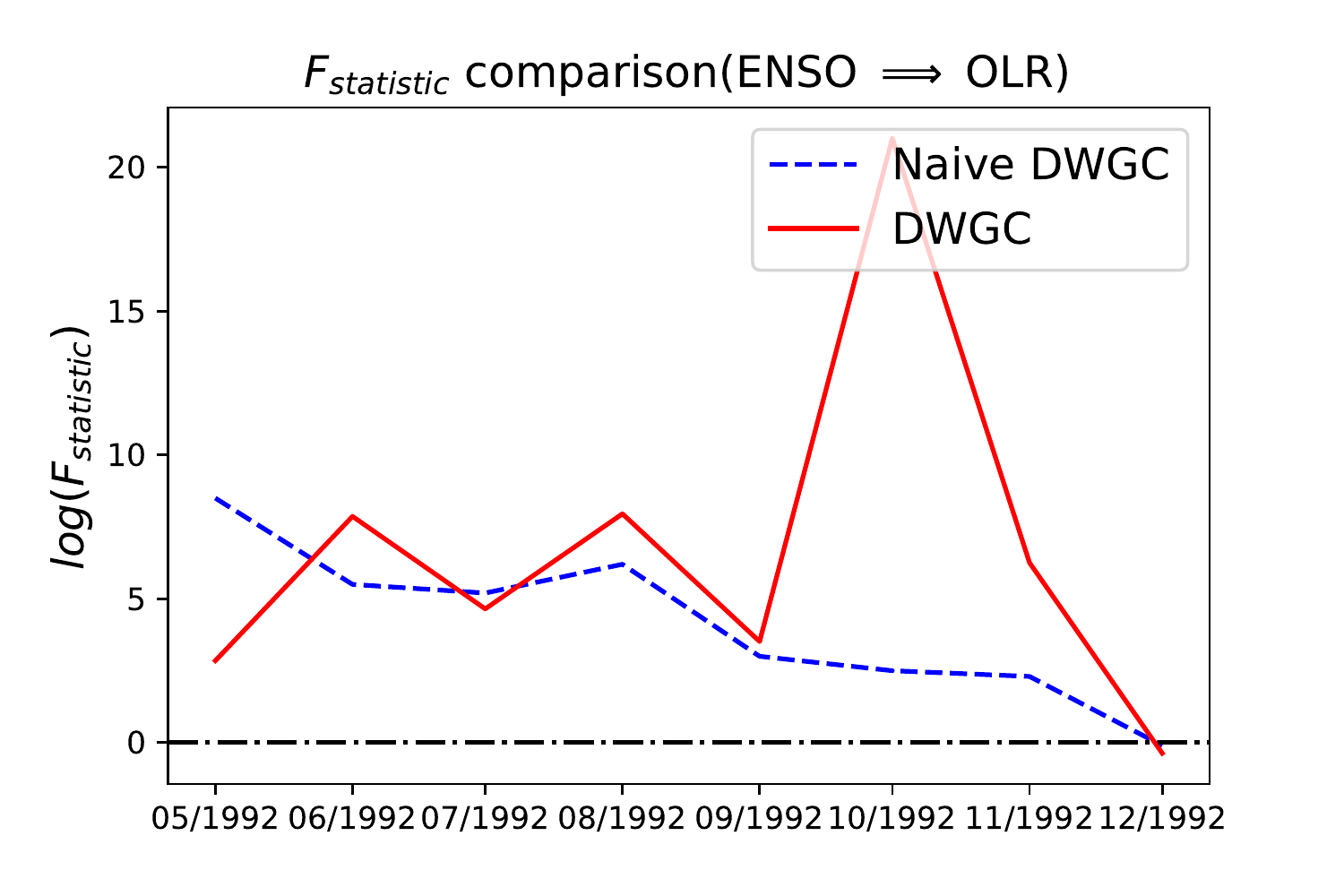}
		\end{minipage}
		\caption{$F_{statistic}$ from ENSO to MKE/OLR using GC(naive DWGC) and DWGC method. In GC, there is a causal relationship between ENSO and MKE \& OLR in most months, but there is no obvious trend of causal change on the whole. In our method DWGC, on the basis of detecting the causal relationship $\text{ENSO} \Longrightarrow \text{MKE}$ and $\text{ENSO} \Longrightarrow \text{OLR}$, we further find that both the causal relationship in autumn and winter(Oct-Dec) is obviously stronger than that in spring and summer(May-Aug)(confirms the priori meteorology conclusion\cite{Kumar1999On}).}
		\label{ENSOfigure1}
	\end{figure}
	\begin{figure}
		\centering
		\begin{minipage}[c]{0.5\textwidth}
			\centering
			\includegraphics[height=4.5cm,width=6cm]{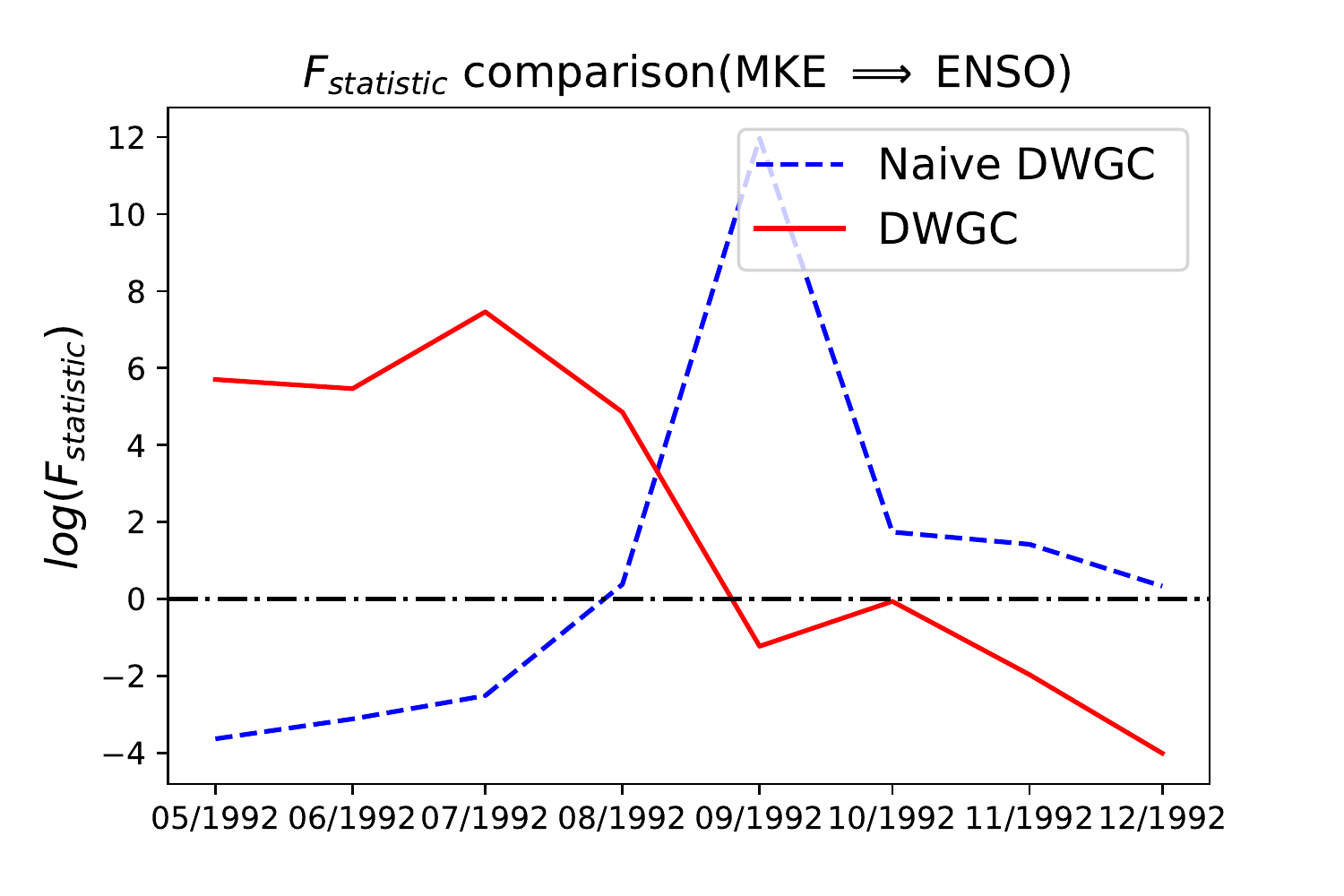}
		\end{minipage}%
		\begin{minipage}[c]{0.5\textwidth}
			\centering
			\includegraphics[height=4.5cm,width=6
			cm]{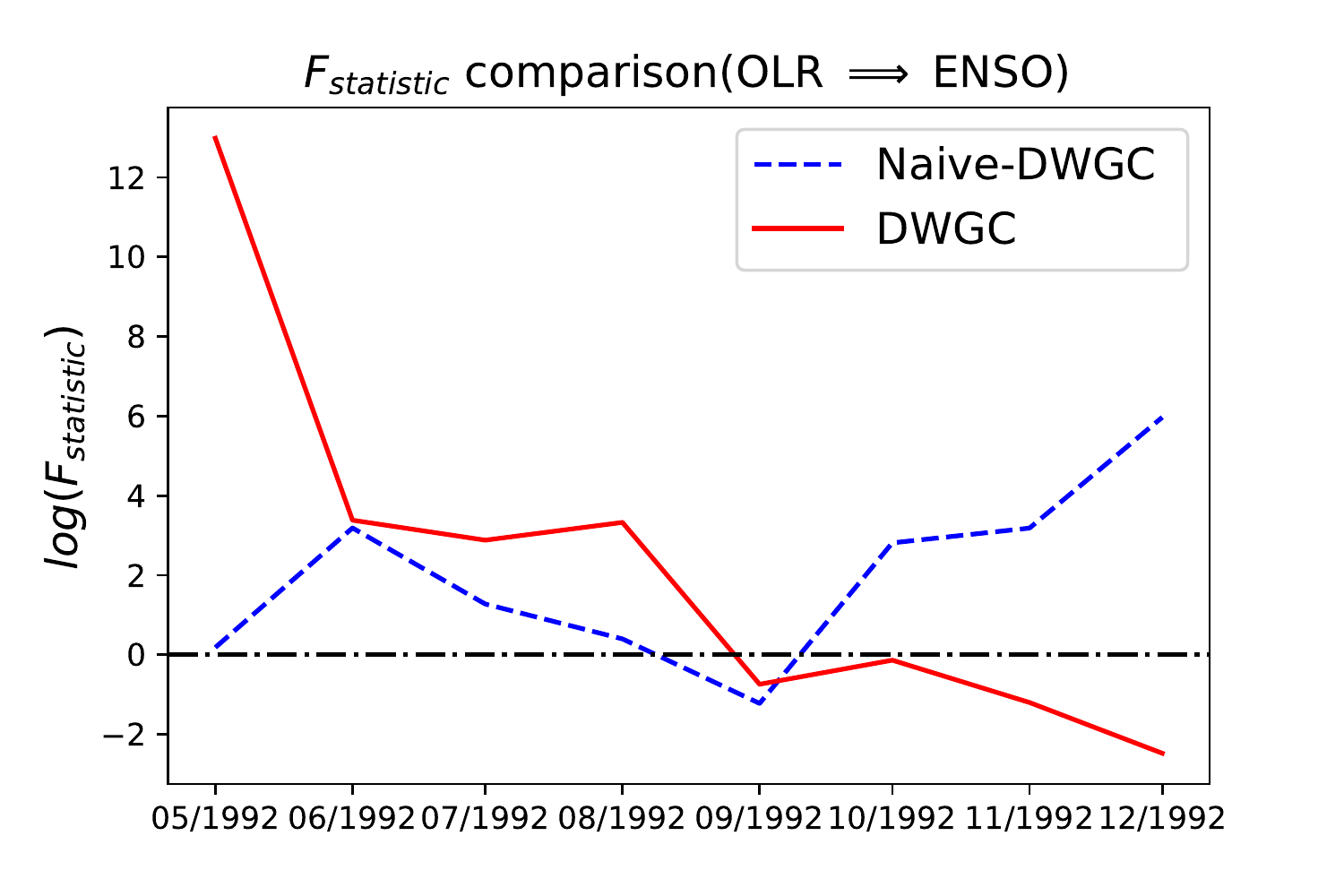}
		\end{minipage}
		\caption{$F_{statistic}$ from MKE/OLR to ENSO using GC(naive DWGC) \& DWGC method. Compared with GC method, DWGC can successfully detect reverse causality $ENSO \Longrightarrow MKE$ in spring and summer(May)(confirms the prior meteoroogy\cite{yasunari1990impact}).}
		\label{ENSOfigure2}
	\end{figure}
	\begin{comment}
	\begin{table}
	\centering
	\fontsize{8}{10}\selectfont    %{字体尺寸}{行距}
	\begin{tabular}{l|cccccccc}
	\toprule
	\diagbox  {method}{month} & May & Jun & Jul & Aug & Sep & Oct & Nov & Dec  \\
	\hline
	GC(ENSO $\Longrightarrow$ MKE) &1.1122	&5.4905	&1.5254	&1.6724	&1.4365	&-2.4351	&0.4243	&1.7107 \\
	DWGC(ENSO $\Longrightarrow$ MKE) &-6.1477   & 5.1326   & 1.4357    &1.5731   & 1.4600   & 4.2516   &-0.2601  &20.3885\\
	GC(ENSO $\Longrightarrow$ OLR) &16.7650   &14.1504  & 13.6686   &16.3217   &16.9633   &19.3671   &22.0062  &16.7119\\
	DWGC(ENSO $\Longrightarrow$ OLR) &2.6127   &3.0435  & 2.9544   &5.2323  &6.0655   &2.5433  &16.3212  &27.5436\\
	\bottomrule
	\end{tabular}\vspace{0cm}
	\caption{$log(F_{statistic})$ from ENSO to Monsoon Kinetic Energy/OLR. We successfully detect that ENSO has a more obvious causal effect on the Asian monsoon in autumn and winter. In spring and summer, the causal effect would be weakened or even reversed.} \label{table}
	\end{table}
	\end{comment}
	
	%In Fig\eqref{ENSOfigure1}\eqref{ENSOfigure2} ,
	We use our model(DWGC) to judge the window-level causalities between ENSO and Asian monsoon in every month. In Fig.~\ref{ENSOfigure1}, we show the F-statistic values of every month of DWGC and naive DWGC without causality indexing for two series (ENSO,MKE) and (ENSO, OLR). The x-axis is the month-time and the y-axis is the log-scale F-statistic value. The black dot-dash line with value zero is the F-statistic threshold.  For naive DWGC without causality indexing, the causality between ENSO and two parameters of east Asian monsoon(MKE,OLR) can be successfully detected, but the difference of causality between May-Aug(spring and summer) and Aug-Dec(autumn and winter) is not significant. However, in our DWGC method, after detecting the basic causal relationship of ENSO $\Longrightarrow$ MKE and ENSO $\Longrightarrow$ OLR, we further find that the causal relationship in autumn and winter is more significant than that in spring and summer in both of them with the larger F-statistic values. This significant causality variation detected by our DWGC is consistent with the academic knowledge\cite{Kumar1999On}.
	
	In Fig.~\ref{ENSOfigure2}, we show the F-statistic values of every month of DWGC and naive DWGC for two series (MKE,ENSO) and (OLR,ENSO), with the same figure axis and F-statistics threshold of Fig.~\ref{ENSOfigure1}. For naive DWGC without causality indexing, we detect the causality $\text{MKE} \Longrightarrow \text{ENSO}$ in Sep-Dec(autumn and winter) and the causality $\text{OLR} \Longrightarrow \text{ENSO}$ in May-Jul(spring and summer) and Oct-Dec(autum and winter). However, in our DWGC method, we detect the causality $\text{MKE} \Longrightarrow \text{ENSO}$ and $
	\text{OLR} \Longrightarrow \text{ENSO}$ in May-Aug(spring and summer). Compared to the naive DWGC method, our causalities results of DWGC are more close to the academia knowledge\cite{yasunari1990impact}.

	\section{Conclusions and Future Work}
	\noindent In this paper, a new task is proposed to detect the window-level dynamic causal relationship between the time series data. 
	By directly conducting the F-test on comparing the window level forecasting predictions with/without the cause channel, the naive DWGC can detect the window-level causalities. This naive DWGC is a special case of the traditional Granger method and is not accurate enough especially when the sliding window length is not large enough.
	We introduce a technique called ``causal indexing'' to reweight the original time series. The purpose of this technique is to decrease the effects of the auto-correlation noise and increase the cross-correlation causal effects. The improved DWGC method is proved to have better causality detection accuracies. causalities detection accuracies.
	As far as we know, this paper is the first to propose and solve the new task of dynamic Granger causal relationship detection at the window level.
	In the experiments on two synthetic and one real datasets, we show that our DWGC method outperforms the traditional GC and the naive DWGC in sense of the causality detection accuracies.
	%with the improved NAR method, which makes the temporal inference results more dynamic, refined and real-time.
	
	The dynamic causalities detected by DWGC is restricted to the same sliding window of every time series. In the future, we are interested in detecting the dynamic causalities without prior knowledge of the sliding window length and between the different sliding windows, and further analyzing the effect of step length for the dynamic causality detection.

	\section*{Acknowledgement}
	%\noindent Code: \url{https://github.com/ZHzhang01/DWGC}
	This work is done when the authors were working at RealAI.

	\bibliography{references} 
	\bibliographystyle{unsrt}
	
	\newpage
	\appendix
	\iffalse
	\section{Anonymous Code and Data Share}
	\label{sec:code}
	We plan to share the anonymous code and data release a couple of weeks after the ECML-2020 deadline. The link is \url{https://drive.google.com/open?id=1H8ez8VncabIKesd8LQkVT-EXVMojaZ2E}.
	\fi
	
	\section{The Simplification of F-test result}\label{eq:lemma 1}
	\begin{lemma}
		If k is large enough, $\sigma_0 = \sigma_0(\hat Y_{t}^{}-Y_{t}^{real})$, we have the probability density function: 
		\begin{equation}
		f_{\frac{{\sum_{t}^{k}\gamma_{t}(\hat Y^{}_{t}-Y^{real}_{t})}}{\sum_{t}^{k}\gamma_t^2}}(x) = O(e^{-\frac{e^6 k}{2(\sigma_0^2+1)}x^2})
		\end{equation}
	\end{lemma}
	\begin{proof}
		Due to that the sum and product of Gaussian variables still satisfy Gaussian distribution N(0,$\sigma_0$)(ignore the systematic error). We have:
		\begin{equation}
		\begin{aligned}
		s_1 = \sum_{i=1}^{k}\gamma_{i} (\hat Y_{}-Y_{real}) &\sim \frac{1}{\sqrt{2\pi k(\sigma_0^2+1)}}e^{-\frac{s_1^2}{2k(\sigma_0)^2+1)}}\\
		s_2 = \sum_{i=1}^{k}\gamma_i^2 &\sim \frac{1}{2^{\frac{k}{2}}}e^{-\frac{s_2}{2}}s_2^{\frac{k}{2}-1}\\
		s3 = \frac{ \sum_{i=1}^{k}\gamma_{i} (\hat Y_{}-Y_{real}) }{\sum_{i=1}^{k}\gamma_i^2} &\sim \int_{0}^{+\infty}\frac{1}{2^{\frac{k}{2}}\tau(\frac{k}{2})\sqrt{2\pi k (\sigma_0^2+1)}}e^{-\frac{t^2 s_3^2}{2k(\sigma_0^2+1)}-\frac{t}{2}}t^{\frac{k}{2}-1}dt
		\end{aligned}
		\end{equation}
		So

		\begin{equation}
		\begin{aligned}
		f(s_3) &=\int_{0}^{+\infty} \sum_{n=0}^{+\infty} \frac{1}{2^{\frac{k}{2}}\tau(\frac{k}{2})\sqrt{2\pi k (\sigma_0^2+1)}}\frac{(-\frac{t^2 x^2}{2k(\sigma_0^2+1)})^n}{n!} e^{-\frac{k}{2}} t^{\frac{k}{2}-1}  dt\\
		&=\sum_{i=0}^{+\infty}\frac{(-1)^n x^{2n}}{2^{\frac{k}{2}} \tau(\frac{k}{2})\sqrt{2\pi k (\sigma_0^2+1)} n!2^n k^n (\sigma_0^2+1)^n} e^{-\frac{t}{2}} t^{\frac{k}{2}-1+2n} dt\\
		&=\sum_{i=0}^{+\infty}\frac{(-1)^n 2^n x^{2n} \tau(2n+\frac{k}{2})}{\tau(\frac{k}{2})\sqrt{2\pi k (\sigma_0^2+1)}n!k^n (\sigma_0^2+1)^n}\\    &\sim \frac{1}{\sqrt{2\pi (\sigma_0^2+1)}} \sum_{n=0}^{+\infty}\frac{(-\frac{e^6 k}{2(\sigma_0^2+1)})^n x^{2n})}{n!}\\
		&\sim \frac{1}{\sqrt{2\pi (\sigma_0^2+1)}} e^{-\frac{e^6 k}{2(\sigma_0^2+1)}x^2}
		\end{aligned}
		\end{equation}
		Its convergence rate is significantly faster than that of the Gaussian variable in the same form(the rate is $e^6 k^2$), which can be ignored as it tends to zero in most cases.
	\end{proof}
	\begin{comment}
	\section{proof of lemma 2}
	Due to $f^k_{\sum_{t=1}^{k}(\hat Y_{t}^{test}-Y_{t}^{real})^2/k}(x) = \frac{k}{2\sigma}Gamma(\frac{kx}{2\sigma}|\frac{k}{2})$:
	\begin{equation}
	\begin{aligned}
	E(F_{statistic}) &= E(\frac{\sum_{t=1}^{k}(\hat Y_t^{}-Y_t^{real})^2}{\sum_{t=1}^{k}(\hat \hat Y_t^{}-Y_t^{real})^2})\\ &=\epsilon_0 E\left((\hat Y^{}-Y^{real})^2\right)E(\frac{1}{(\hat Y^{}-Y^{real})^2}) \\ &=\epsilon_0 (k\epsilon_0 \sigma)(\frac{1}{(k-2)\epsilon_0 \sigma}) \approx \epsilon_0.
	\end{aligned}
	\end{equation}$\hfill\blacksquare$ 
	\end{comment}
	\section{Proof of Theorem 1}\label{proof of theorem 1}
	\subsection{Preliminary}
	Nash efficiency coefficient\cite{nash1970river} is often used to evaluate the performance of simulation prediction, which is expressed as:
	$$Nash = 1-\frac{\sum_{t=1}^{n}(\hat Y^{}_t-Y_t)^2}{\sum_{t=1}^{n}(Y_t-E(Y_t))^2},Nash_{\Phi} = 1-\frac{\sum_{t=1}^{n}(\hat Y^{\Phi}_t-\phi Y_t)^2}{|E(\phi)|^2\sum_{i=1}^{n}(Y_t-E(Y_t))^2}$$
	if out model's Nash efficiency coefficient is stable, that is: $Nash = Nash_{\Phi}$, we can get:
	$$\sum_{t=1}^{n}(\hat Y^{\phi}_t(\hat Y^{\phi}_{t|j})-\Phi Y_t)^2 = |E(\phi)|^2 \sum_{t=1}^{n}(\hat Y_t^{}(\hat Y^{}_{t|j})-Y_t)^2,$$
	\subsection{}
	\begin{proof}
		\begin{equation}
		f_{F_{statistic}}^k (\epsilon) = \sum_{i=0}^{\frac{k}{2}-1}(\frac{k}{2\sigma \sqrt{\epsilon_0}})^k \frac{C_{\frac{k}{2}-1}^{i}(\epsilon-1)^i\epsilon^{\frac{k}{2}-1-i}   \frac{(\frac{k}{2})^{\frac{k}{2}}\tau(\frac{k}{2}+i)}{\tau(\frac{k}{2})[\frac{k}{2}(1+\frac{\epsilon-1}{\epsilon_0 \sigma})]^{\frac{k}{2}+i}}  (\frac{2\epsilon_0 \sigma}{k(\epsilon+\epsilon_0)})^{k-1-i}\tau(k-1-i)}{\tau(\frac{k}{2})^2}\label{27}
		\end{equation}
		Let's call each series $g(\epsilon,k,i)$, so $f_{F_{statistic}}^{k} = \sum_{i=0}^{\frac{k}{2}-1}g(\epsilon,k,i)$. We recursively prove that $f_{F_{statistic}}^{k}$ is a monotonic function with respect to k(when k is greater than a certain value,$\epsilon>1$):
		\begin{equation}
		\begin{aligned}
		\frac{d (ln(g(\epsilon,k,i)))}{dk} = &lnk+1-ln(2\sigma \sqrt{\epsilon_0})-\digamma(\frac{k}{2})-\frac{1}{2}\digamma(\frac{k}{2}-i)+\frac{1}{2}ln\epsilon+\frac{1}{2}\digamma(\frac{k}{2}+i)-\\&\frac{1}{2}ln(1+\frac{\epsilon-1}{\epsilon_0 \sigma})  -\frac{i}{k}+ln(\frac{2\epsilon_0 \sigma}{\epsilon+\epsilon_0})+\digamma(k-1-i)-lnk-\frac{k-1-i}{k}\\
		=&ln(\frac{\sqrt{\epsilon_0}  \sqrt{\epsilon}}{(\epsilon+\epsilon_0)\sqrt{1+\frac{\epsilon-1}{\epsilon_0 \sigma}}})+\frac{1}{k}+ln(\frac{2\sqrt{\frac{k}{2}+i}}{\sqrt{\frac{k}{2}-i}}\frac{(k-1-i)}{k})
		\end{aligned}
		\end{equation}
		For each $\epsilon >1$, if we can find $\alpha$, $i \in (\alpha k-1, \frac{1}{2}k-1)$, we have \\
		\begin{equation}
		\frac{2\sqrt{\epsilon_0 \epsilon }}{(\epsilon+\epsilon_0)\sqrt{1+\frac{\epsilon-1}{\epsilon_0 \sigma}}}\frac{\sqrt{\frac{1}{2}+\alpha}}{\sqrt{\frac{1}{2}-\alpha}}(1-\alpha)>0, f_{statistic}^k(\epsilon) -\sum_{i=0}^{\alpha k -1}g(\epsilon,k,i)=o(1)
		\end{equation}
		We can get $\frac{d(ln(g(\epsilon,k,i)))}{dk}>0$, so
		\begin{equation}
		g(\epsilon,k+2,i)>g(\epsilon,k+1,i)>g(\epsilon,k,i), \epsilon  >1
		\end{equation}
		Here comes to the conclusion:
		\begin{equation}
		f_{F_{statistic}}^k=\sum_{i=0}^{\frac{k}{2}-1}g(\epsilon,k,i)<\sum_{i=0}^{\frac{k}{2}}g(\epsilon,k,i)<\sum_{i=0}^{\frac{k}{2}}g(\epsilon,k+2,i)=f_{F_{statistic}}^{k+2} 
		\end{equation}
	\end{proof}
	\noindent Also, we can give a simpler proof:
	\begin{proof}
		\begin{equation}
		\begin{aligned}
		&P(F_{statistic}>1) \\
		=&P(\sum_{t}{(\hat{Y}_{it}-Y^{\text{real}}_{it})}^2 > \sum_{t}{(\hat{Y}_{it|j}-Y^{\text{real}}_{it})}^2)\\
		=& 1-\sum_{m=0}^{+\infty} \frac{\Gamma(k+m) (\frac{1}{\epsilon_0})^{\frac{k}{2}+m}}{(\frac{1}{\epsilon_0}+1)^{k+m}\Gamma(\frac{k}{2})\Gamma(\frac{k}{2}+m+1)}\label{eqn:second proof}
		\end{aligned}
		\end{equation}
		Let's take the derivative of the next term:
		\begin{equation}
		\begin{aligned}
		&{( \sum_{m=0}^{+\infty} \frac{\Gamma(k+m) (\frac{1}{\epsilon_0})^{\frac{k}{2}+m}}{(\frac{1}{\epsilon_0}+1)^{k+m}\Gamma(\frac{k}{2})\Gamma(\frac{k}{2}+m+1)})}^{'} \\
		=&\phi(k+m)+\frac{1}{2}ln\frac{1}{\epsilon_0}-ln(\frac{1}{\epsilon_0}+1)-\phi(\frac{k}{2})-\phi(\frac{k}{2}+m+1)\\
		&<0
		\end{aligned}
		\end{equation}
		Therefore, equation \eqref{eqn:second proof} shows a single increasing trend with the increase of window length k.
	\end{proof}
	\section{Proof of Theorem 2}\label{proof-theorem-2}

	%\section{Proof of theorem 3}
	\begin{lemma}
		If ${\gamma_t}$ are k independent, circular symmetric complex Gaussian random variables with mean 0 and variance $\phi_t^2$, we have:
		\begin{equation}
		f_{\sum_{t=1}^{k}(\gamma_t^2)}\left(x ; k, \phi_{1}^{2}, \ldots, \phi_{k}^{2}\right)=\sum_{t=1}^{k} \frac{e^{-\frac{x}{\phi_{t}^{2}}}}{\phi_{t}^{2} \prod_{j=1, j \neq t}^{k}\left(1-\frac{\phi_{j}^{2}}{\phi_{t}^{2}}\right)} \text { for } x \geq 0
		\end{equation} \\
	\end{lemma}
	\begin{theorem}
		for each k, consider the window $\{\phi_1,\phi_2,..\phi_k\}$, \textbf{the sufficient condition} for $\int_{1}^{+\infty}f_{F_{statistic}}^{k,\Phi}(\epsilon)>\int_{1}^{+\infty}f^{k}_{F_{statistic}}(\epsilon)$ is:
		\begin{equation}
		\sum_{q=1}^{\frac{k}{2}} \frac{1}{\hat \phi_q^2\prod_{j \neq q}(1-\frac{\hat \phi_j^2}{\hat \phi_q^2})}>\frac{\tau(k-1)}{\tau(\frac{k}{2})^2}{(\frac{|E(\Phi)|^2}{\phi_t^2})}_{max}, \hat \phi_{m}=\phi_{2m}=\phi_{2m-1} \label{final-condition1}
		\end{equation}
	\end{theorem}
	\begin{proof}
		The $F_{statistic}$ can be represented as follows: 
		\begin{equation}
		\begin{aligned}
		F_{statistic} \approx \frac{\sum_{t=1}^{k}{(\hat Y_t^{\Phi}-\phi_t Y_t^{real})}^2+\sum_{t=1}^{k} (\phi_t \gamma_{t})^2 }{\sum_{t=1}^{k}{( \hat Y_{t|j}^{\Phi}-\phi_t Y_t^{real})}^2+\sum_{t=1}^{k} (\phi_t\gamma_{t})^2 }\label{21}
		\end{aligned}
		\end{equation}
		\begin{comment}
		So the equation(16) can be re-written as :
		\begin{equation}
		\begin{aligned}
		& G(\epsilon,k)=\\
		& = \int f_{\phi \gamma_i^2}(X_2)\int _{0}^{+\infty}\hat f_{\hat x1}(t)[1-F_{x1}[(\epsilon-1)x2+\epsilon t]] dt dX_2 \\
		& = E_{l=1}^{ \lfloor \frac{m}{k}\rfloor}\int_{0}^{+\infty} \frac{k}{2}\sum_{i=1}^{\frac{k}{2}}\frac{e^{-\frac{kx}{2\phi_i^2}}}{\prod_{i \neq j}(\phi_i^2-\phi_j^2)}* \\
		& \int_{0}^{+\infty} {(\frac{k}{2\pi K_{\hat {X1}}^{''}(m_1)})}^{\frac{1}{2}}e^{k(K_{\hat   X_1}(m_1)-m_1t  )}
		{(\frac{k}{2\pi K_{ {X1}}^{''}(m_2)})}^{\frac{1}{2}}e^{ k(K_{ X_1}(m_2)-((\epsilon-1)X_2+\epsilon *t)m_2)}
		dt dX_2\\
		& = \sum_{q=1}^{\frac{k}{2}}\sum_{i=0}^{\frac{k}{2}-1}\frac{\tau(i+1)2^{k-4-i}}{|\Phi|^{4i+4}\phi_q^2\prod_{j \neq q}(1-\frac{\phi_j^2}{\phi_q^2})(\frac{k}{2\phi_q^2}+\frac{k(\epsilon-1)}{2|\Phi|^4\sigma^2 \epsilon_0^2})^{i+1}}* \frac{(\epsilon-1)^ie^k \epsilon_0^{k-2-2i}\tau(k-1-i)C_{\frac{k}{2}-1}^{i}\epsilon^{\frac{k}{2}-1-i}}{\pi k^{k-3-i}(\epsilon_0^2+\epsilon)^{k-1-i}\sigma^{2i+2}}
		\end{aligned}
		\end{equation}
		where $K^{'}(t) = \sum_{r}^{k} \frac{ \phi_r}{1-2\phi_r t} = x2$. Compared with equation (15), when we adjust $\Phi$ by KL distance in (7), the sufficient condition for us to expect the rate of causality detection to increase is:
		\end{comment}
		According to \eqref{21}, two influencing factors-prediction error and noise-have changed.
		\begin{comment}
		\begin{enumerate}
		\item For the prediction error, due to assumption 1, different with lemma 2:$$f^k_{\sum_{t=1}^{k}( Y_{t}^{test,\Phi}-\phi_t Y_{t}^{real})^2}(x)\sim \frac{1}{2\epsilon_0 \sigma E(\phi)^2}Gamma(\frac{x}{2\epsilon_0\sigma E(\phi^2)}|\frac{k}{2})$$$$f^k_{\sum_{t=1}^{k}(\hat Y_{t}^{test,\Phi}-\phi_t Y_{t}^{real})^2}(x)\sim \frac{1}{2 \sigma E(\phi)^2} Gamma(\frac{x}{2\sigma E(\phi^2)}|\frac{k}{2})$$.
		\item For the noise, due to lemma 3, different with equation \eqref{12}, the adjacent time points are given the same weight.  the density function of noise is: $f_{\sum(\phi_t \gamma_t)^2}(x) = \sum_{t=1}^{\frac{k}{2}} \frac{e^{-\frac{x}{\hat \phi_{t}^{2}}}}{\hat \phi_{t}^{2} \prod_{j=1, j \neq t}^{\frac{k}{2}}\left(1-\frac{\hat \phi_{j}^{2}}{\hat \phi_{t}^{2}}\right)}, \hat \phi_{m}=\phi_{2m}=\phi_{2m-1} $.
		\end{enumerate}
		\end{comment}
		Using lemma 2, assume $\phi_{2m}=\phi_{2m-1}$, we do the series expansion $f^{k,\Phi}_{F_{statistic}}(\epsilon) = \sum_{q=1}^{\frac{k}{2}} \sum_{t=0}^{\frac{k}{2}-1} g_2(\epsilon,k,t,\phi_q)$,
		so the sufficient condition for $\int_{\epsilon_{min}}^{\epsilon_{max}}f_{F_{statistic}}^{k,\Phi}(\epsilon)>\int_{\epsilon_{min}}^{\epsilon_{max}}f^{k}_{F_{statistic}}(\epsilon)$ is: %补充说出那个chongfen条件，并和lemma2 并列对比
		\begin{equation}\label{final-condition2}
		\sum_{q=1}^{\frac{k}{2}} g_2(\epsilon,k,t,\phi_q) > g(\epsilon,k,t), 
		\end{equation}
		Comparing the series expansion of naive DWGC and DWGC, the sufficient condition of \eqref{final-condition2} is:

		$\forall k,q:$
		\begin{equation}
		\sum_{q=1}^{\frac{k}{2}}\frac{k\tau(i+1)}{2 \phi_q^2\prod_{j \neq q}(1-\frac{\phi_j^2}{\phi_q^2})(\frac{k}{2\phi_q^2}+\frac{k(\epsilon-1)}{2E|(\Phi)|^2\sigma \epsilon_0})^{i+1}|E(\Phi)|^{2i+2}} > \frac{(\frac{k}{2})^{\frac{k}{2}}\tau(\frac{k}{2}+i)}{(\frac{k}{2}+\frac{k(\epsilon-1)}{2\sigma\epsilon_0})^{\frac{k}{2}+i}\tau(\frac{k}{2})}
		\end{equation}
		The sufficient condition can be converted to: :\\
		\begin{equation}
		\forall q, \sum_{q=1}^{\frac{k}{2}}\frac{k\tau(i+1)}{2\phi_q^2\prod_{j \neq q}(1-\frac{\phi_j^2}{\phi_q^2})(\frac{k|E(\Phi)|^2}{2\phi_q^2}+\frac{k(\epsilon-1)}{2\sigma \epsilon_0})^{i+1}} > \frac{(\frac{k}{2})^{\frac{k}{2}}\tau(\frac{k}{2}+i)}{(\frac{k}{2}+\frac{k(\epsilon-1)}{2\sigma\epsilon_0})^{\frac{k}{2}+i}\tau(\frac{k}{2})} 
		\end{equation}
		\begin{equation}
		\forall q, \sum_{q=1}^{\frac{k}{2}}\frac{k\tau(i+1)}{2\phi_q^2\prod_{j \neq q}(1-\frac{\phi_j^2}{\phi_q^2})} > \frac{(\frac{k}{2}(\frac{(|E(\Phi)|^2}{\phi_q^2})_{max}+\frac{k(\epsilon-1)}{2\sigma \epsilon_0})^{i+1}(\frac{k}{2})^{\frac{k}{2}}\tau(\frac{k}{2}+i)}{(\frac{k}{2}+\frac{k(\epsilon-1)}{2\sigma\epsilon_0})^{\frac{k}{2}+i}\tau(\frac{k}{2})}\\
		\end{equation}
		\begin{equation}
		\forall q, \sum_{q=1}^{\frac{k}{2}}\frac{k\tau(i+1)}{2\phi_q^2\prod_{j \neq q}(1-\frac{\phi_j^2}{\phi_q^2})} > \frac{(\frac{k}{2}(\frac{(|E(\Phi)|^2}{\phi_q^2})_{max}+\frac{k(\epsilon-1)}{2\sigma \epsilon_0})(\frac{k}{2})^{\frac{k}{2}}\tau(\frac{k}{2}+i)}{(\frac{k}{2}+\frac{k(\epsilon-1)}{2\sigma\epsilon_0})^{\frac{k}{2}}\tau(\frac{k}{2})}\\
		\end{equation}
		\begin{equation}
		\begin{aligned}
		\forall q, \sum_{q=1}^{\frac{k}{2}}\frac{k}{2\phi_q^2\prod_{j \neq q}(1-\frac{\phi_j^2}{\phi_q^2})} >& \frac{(\frac{k}{2}(\frac{(|E(\Phi)|^2}{\phi_q^2})_{max})(\frac{k}{2})^{\frac{k}{2}}\tau(\frac{k}{2}+i)}{(\frac{k}{2})^{\frac{k}{2}}\tau(\frac{k}{2})\tau(i+1)}\\=&\frac{\tau(k-1)}{\tau(\frac{k}{2})^2} {(\frac{|E(\Phi)|^2}{\phi_q^2})}_{max} \\\label{35}
		\end{aligned}
		\end{equation}
		In all, the sufficient condition of improving the naive DWGC method is:
		\begin{equation}
		\begin{cases}
		\forall q, {(\frac{|E(\Phi)|^2}{\phi_q^2})}_{max} = \beta,& \text{}\\
		\sum_{q=1}^{\frac{k}{2}} \frac{1}{\phi_q^2\prod_{j \neq q}(1-\frac{\phi_j^2}{\phi_q^2})}>\frac{\tau(k-1)}{\tau(\frac{k}{2})^2}\beta & \text{}
		\end{cases}\label{final-condition}
		\end{equation}
	\end{proof}
	
	Under this sufficient condition \eqref{final-condition1}, we can add a regularization term to the original loss \eqref{eqn:model_loss}.
	\begin{comment}
	the selection method of $\Phi$ in the model's loss function \eqref{eqn:modell_oss} is a special case of sufficient condition \eqref{final-condition1} in most cases if we set the regularization term as follows(k is large enough):
	\end{comment}
	\begin{comment}
	\begin{equation}
	\begin{aligned}
	&\sum_{i,l}|\phi^{i}_{2l}-\phi^{i}_{2l+1}|+ \alpha_1 (N-|\sum_{m=t}^{t+\frac{k}{2}-1}\phi_m^2|)^{+} + \alpha_2(Q-(\frac{E(\phi^2)}{\phi_{p}^2})_{max})^{+}+ \\ &(k log2-log\sum_{m=t}^{m+\frac{k}{2}-1}\frac{1}{\phi_m^2})^++\sum_{2m \in (t,t+\frac{k}{2}-1)}(\prod_{j \neq 2m}|1-\frac{\phi_j^2}{\phi_{2m}^2}|-1)^+\\ \label{condition28}
	\end{aligned}
	\end{equation}
	$\sum_{i,l}|\phi_{2l}-\phi_{2l+1}|$ aims to prompt two adjacent $\phi_t $'s the same, where $ 2l,2l+1,p \in (t, t+k-1)$, $\alpha_1,\alpha_2$,N,Q are constants.\\
	\end{comment}
	\begin{equation}
	\alpha \sum_{i,l}|\phi^{i}_{2l}-\phi^{i}_{2l+1}|+\beta Relu( \frac{\tau(k-1)}{\tau(\frac{k}{2})^2}{(\frac{|E(\Phi)|^2}{\phi_p^2})}_{max}-\sum_{q=1}^{\frac{k}{2}} \frac{1}{\hat \phi_q^2\prod_{j \neq q}(1-\frac{\hat \phi_j^2}{\hat \phi_q^2})}), \label{condition_reg}
	\end{equation}

	where $\sum_{i,l}|\phi_{2l}-\phi_{2l+1}|$ aims to prompt two adjacent $\phi_t $'s the same, $ 2l,2l+1,p \in (t, t+k-1)$, $\hat \phi_m = \frac{1}{2}(\phi_{2m}+\phi_{2m+1})$. 
	
	At last, We further explain that it is possible to satisfy the inner variable of $Relu(\cdot)$ less than 0, so the loss function can be theratically easy to train after adding the regularization term.
	
	We convert the $Relu(\cdot)$ of the regularization term \eqref{condition_reg} into a more concise form:
	\begin{equation}
	\begin{aligned}
	&\sum_{i,l}|\phi^{i}_{2l}-\phi^{i}_{2l+1}|+ \alpha_1 (N-|\sum_{m=t}^{t+\frac{k}{2}-1}\phi_m^2|)^{+} + \alpha_2((\frac{E(\phi)^2}{\phi_{p}^2})_{max}-Q)^{+}+ \\ &(k log2-log\sum_{m=t}^{m+\frac{k}{2}-1}\frac{1}{\phi_m^2})^++\sum_{2m \in (t,t+\frac{k}{2}-1)}(\prod_{j \neq 2m}|1-\frac{\phi_j^2}{\phi_{2m}^2}|-1)^+\\ \label{condition28}
	\end{aligned}
	\end{equation}
	$\sum_{i,l}|\phi_{2l}-\phi_{2l+1}|$ aims to prompt two adjacent $\phi_t $'s the same, where $ 2l,2l+1,p \in (t, t+k-1)$, $\alpha_1,\alpha_2$,N,Q are constants. We use $[1]\sim[5]$ respectively to represent each item in regularization\eqref{condition28}, and note where to use in the following derivation.
	
	Without prejudice to the problem, we sort the $\Phi_q$ in the above equation:$\phi_1<\phi_2<...<\phi_{\frac{k}{2}}$, and $\frac{k}{2}$ is an even number, in equation \eqref{final-condition},if $q \in {2k+1,k=0,1,2,...}$, $\frac{1}{\Phi_q^2\prod_{j \neq q}(1-\frac{\phi_j^2}{\phi_q^2})}<0$, if $q \in {2k,k=1,2,3...}$, $\frac{1}{\Phi_q^2\prod_{j \neq q}(1-\frac{\phi_j^2}{\phi_q^2})}>0$, 
	So the sufficient condition of equation \eqref{final-condition}:
	\begin{equation}
	\begin{aligned}
	\sum_{q=1}^{\frac{k}{2}} \frac{1}{\phi_q^2\prod_{j \neq q}(1-\frac{\phi_j^2}{\phi_q^2})} &> \sum_{m=1}^{\frac{k}{4}}\frac{1}{\phi_{2m}^2}[5]+\sum_{m=1}^{\frac{k}{4}}\frac{1}{\phi_{2m-1}^2(1-\frac{\sum_q\phi_q^2-\phi_{2m-1}^2}{\phi_{2m-1}^2})}\\&>\sum_{m=1}^{\frac{k}{4}}\frac{1}{\phi_{2m}^2}-\sum_{m=1}^{\frac{k}{4}}\frac{1}{\sum_{i=1}^{\frac{k}{2}}\phi_{i}^2-2\phi_{2m}^2}
	\end{aligned}
	\end{equation}
	
	Here, we have $\phi_{2}^2<\phi_{4}^2<\phi_{6}^2<...<\phi_{2m-4}^2<\frac{1}{3}\sum_{i=1}^{\frac{k}{2}}\phi_i^2$, $\exists \alpha, \phi_{2m-4}^2 = \frac{1-\alpha}{3-2\alpha}\sum_{i=1}^{\frac{k}{2}}\phi_i^2$. The sufficient condition in \eqref{final-condition} can be transformed as:
	\begin{equation}
	\sum_{q=1}^{\frac{k}{2}-4} \frac{1}{\phi_q^2\prod_{j \neq q}(1-\frac{\phi_j^2}{\phi_q^2})} > \sum_{m=1}^{\frac{k}{4}-2}\frac{1}{\phi_{2m}^2}-\sum_{m=1}^{\frac{k}{4}-2}\frac{1}{\sum_{i=1}^{\frac{k}{2}}\phi_{i}^2-2\phi_{2m}^2}> \sum_{m=1}^{\frac{k}{4}-2}\frac{\alpha}{\phi_{2m}^2}\label{38}
	\end{equation}
	\begin{equation}
	\sum_{q=\frac{k}{2}-3}^{\frac{k}{2}} \frac{1}{\phi_q^2\prod_{j \neq q}(1-\frac{\phi_j^2}{\phi_q^2})} =\sum_{m=\frac{k}{4}-1}^{\frac{k}{4}}\frac{1}{\phi_{2m}^2}-\sum_{m=\frac{k}{4}-1}^{\frac{k}{4}}\frac{1}{\sum_{i=1}^{\frac{k}{2}}\phi_{i}^2-2\phi_{2m}^2}>-\frac{2}{\sum_{i=1}^{\frac{k}{2}}\phi^2}
	\end{equation}
	\begin{equation}
	\sum_{q=1}^{\frac{k}{2}} \frac{1}{\phi_q^2\prod_{j \neq q}(1-\frac{\phi_j^2}{\phi_q^2})} >\sum_{m=1}^{\frac{k}{4}-2}\frac{\alpha}{\phi_{2m}^2}-\frac{2}{\sum_{i=1}^{\frac{k}{2}}\phi^2} >\sum_{m=1}^{\frac{k}{4}-2}\frac{\alpha}{\phi_{2m}^2}-\frac{2}{N}[2]
	\end{equation}
	So under the regularization term, we have $\sum\frac{\alpha}{\phi^2}>\frac{\tau(k-1)}{\tau(\frac{k}{2})^2}(\frac{E(\Phi)^2}{\phi_t^2})_{max}[3][4]$ in most cases. Therefore, training regularization term \eqref{condition28} is a sufficient condition for training the original \eqref{condition_reg}, while \eqref{condition28} is easy to train and converge.

\end{document}